\documentclass[10pt,a4paper]{article}
\usepackage[latin1]{inputenc}

\usepackage{amsmath}
\usepackage{amsthm}
\usepackage{amsfonts}
\usepackage{amssymb}
\usepackage{graphicx}
\usepackage{float}
\usepackage{color}
\usepackage{subfig}
\usepackage{authblk}
\usepackage{natbib}
\usepackage[margin=3cm,top=3cm,bottom=3cm]{geometry}

\newtheorem{proposition}{Proposition}

\title{{Smiles all around:}\\ FX joint calibration in a multi-Heston model}

\author[1]{Alvise De Col}
\author[2]{Alessandro Gnoatto}
\author[3,4]{Martino Grasselli}
\affil[1]{\small UBS-AG, UBS, Europastrasse 1, 8152 Opfikon, (Switzerland)}
\affil[2]{\small Mathematisches Institut, LMU M\"unchen (Germany)}
\affil[3]{\small Dipartimento di Matematica, Padova (Italy)}
\affil[4]{\small D\'{e}partement Math\'ematiques et Ing\'enierie
Financi\`{e}re, ESILV, Paris La D\'efense (France)}

\begin{document}
\maketitle
\begin{abstract}
We introduce a novel multi-factor Heston-based stochastic
volatility model, which is able to reproduce consistently typical
multi-dimensional FX vanilla markets, while retaining the
(semi)-analytical tractability typical of affine models and
relying on a reasonable number of parameters. A successful joint
calibration to real market data is presented together with various
in- and out-of-sample calibration exercises to highlight the
robustness of the parameters estimation. The proposed model
preserves the natural inversion and triangulation symmetries of FX
spot rates and its functional form, irrespective of choice of the
risk-free currency. That is, all currencies are treated in the same way.
\\
\end{abstract}

\section{Introduction}

The FX OTC market keeps on growing at an unabated speed. In less
than ten years the volume of FX spot, forward and option
transactions has more than tripled, reaching almost 4 trillion USD
daily turnover in 2010, see \cite{BIS2010}. 
From a modeling perspective, capturing the global nature of the FX
option market, is a highly non trivial task. While the modeling of
a single FX spot underlying in isolation has been thoroughly
analyzed and poses challenges that are similar in nature
with other single-dimensional asset classes, like equities, the
simultaneous representation of multiple FX spot rates is by no
means a straightforward extension. Unlike other asset classes,
both inversions and appropriate multiplications/divisions of FX
rates are tradable FX cross rates (eg, EUR/JPY can be derived from EUR/USD
and USD/JPY). For cases in which options on the crosses are liquidly
traded, a consistent model of multiple FX rates must be able to
reprice vanilla options that are written on different FX
underlying rates.

Let $S^\mathrm{d,f}(t)$ be the spot exchange rate at time $t$ as the amount of domestic (d) currency for one unit of
foreign currency (f). A consistent multi-dimensional FX model must be symmetric with respect to inversion
and triangulation, that is

\begin{itemize}
\item[1.] the flipped process $1/S^{i,l}(t)$ in the foreign (i.e.,
$l$) risk neutral measure follows the same type of process as the
original $S^{i,l}(t)$;
\item[2.] the inferred cross rate
$S^{l,m}(t) = S^{i,m}(t)/S^{i,l}(t)$ follows the same type of
process as the original main currency pairs $S^{i,m}(t)$,
$S^{i,l}(t)$ in its domestic (i.e., $l$) risk neutral measure.
\end{itemize}

In pre-smile times, the standard Black-Scholes model (see
\cite{bla73}, \cite{GarmanKohl}) could be easily extended to
describe multiple FX rates in a consistent fashion
\citep{book_Lipton,bookWystup06,bookClark11}. It is indeed trivial
to show that the inverse of a geometric Brownian motion is still a
geometric Brownian motion and moreover, by choosing the
correlation between the stochastic drivers of the two mains to be
\begin{equation}
\rho_{im-il}(t) = \frac{ \sigma_{il}^2(t) + \sigma_{im}^2(t)- \sigma_{lm}^2(t)}
                                                  {2\sigma_{il}(t)\sigma_{im}(t)},
\end{equation}
the cross rate $S^{l,m}(t)$ follows a geometric Brownian motion
with implied volatility $\sigma_{lm}^2(t)$ ($\sigma_{il}^2(t)$ and
$\sigma_{im}^2(t)$ are the implied volatilities of the two mains).

If we include the effect of volatility smiles, the extension from a single-dimensional model to a
multi-dimensional one is way less trivial. Beside preserving the specific FX symmetries, a suitable
model must be able to fit jointly not just the market volatility smiles of the main currency pairs, say
$S^{i,m}(t)$ and $S^{i,l}(t)$, but also the smile of the cross rate $S^{l,m}(t) = S^{i,m}(t)/S^{i,l}(t)$.
In case more than three currencies are present, say $N$, the model should be able to fit simultaneously
the market volatility smiles of $N-1$ mains and $(N-1)(N-2)/2$ crosses,
a task which is increasingly demanding with growing $N$.

In this paper we introduce a  multi-dimensional stochastic
volatility model that is based on the \cite{Heston93} model and is
able to satisfy the inversion and triangulation symmetries, while
being able to produce a satisfactory joint calibration of main and
cross implied volatility smiles. The choice of a stochastic
volatility model is consistent with the persistency of the
volatility smile effect over different maturities, indicating that FX spot returns are not normally distributed, as
observed in time-series analysis \citep{carrwu07}. Analysis of
butterflies and risk reversals times series otherwise show a
stable behavior for the excess kurtosis of the implied
risk-neutral distribution across currencies and expiries, whereas
risk-reversals appear to vary considerably over time, so that even
sign changes are present. This translates into an empirical
distribution whose skewness changes significantly over time.

The inversion symmetry is generally not satisfied by stochastic volatility models. Consider for example the popular SABR model
\begin{flalign*}
    &dS(t) = S(t)(r_\mathrm{d} - r_\mathrm{f}) dt + \sigma (t) S^\beta (t) dW(t), \\
    &d\sigma (t) = \alpha \sigma(t) dZ(t), \quad\quad d\langle W, Z\rangle_t = \rho dt.
\end{flalign*}
It can be shown by straightforward calculations that the inverted SABR process ($\hat S(t) = 1/S(t)$) in the foreign risk-neutral measure $\mathbb{Q}_f$ reads
\begin{flalign*}
    &d\hat S(t) = \hat S(t)(r_\mathrm{f} - r_\mathrm{d}) dt + \sigma (t) \hat S^{2 - \beta}(t) dW^{\mathbb{Q}_f}(t), \\
    &d\sigma (t)= \rho \alpha \sigma^2(t) \hat S^{1-\beta}(t) dt + \alpha \sigma(t) dZ^{\mathbb{Q}_f}(t), \quad\quad d\langle W^{\mathbb{Q}_f}, Z^{\mathbb{Q}_f}\rangle_t = -\rho dt.
\end{flalign*}
The inverted process is not of SABR type. Even for $\beta = 1$
(log-normal case) the stochastic volatility process has an
addition drift term proportional to the spot-volatility
correlation, spoiling the inversion symmetry. Other examples that
do not satisfy the inversion symmetries are the GARCH (e.g.\ see
\cite{lewis2000}) and the \cite{scott_stcev} models (i.e., popular
stochastic volatility process whose instantaneous volatility is
proportional to the exponential of a Ornstein-Uhlenbeck process).
The Heston model, however, naturally satisfies this symmetry, see
also \cite{note_dbr}.

The preservation of the triangulation symmetry depends on the
specification of the intra-currency pair correlation structure. If
we consider the simplest and often standard choice of constant
correlation, it is easy to show that the FX cross rate implied by
the division of two SABR currency pairs follows a process that is
not of SABR type. As long as we keep a constant correlation
between the main currency pairs, also a Heston specification for the stochastic
volatility process of the main FX rates leads to the
definition of a cross currency pair which is functionally
different. Hence, in order to achieve the
triangulation symmetry, one needs to use a different paradigm in
the specification of the correlation.

The model we present in this paper is a multi-factor stochastic
volatility model of \cite{Heston93} type. Multi-factor stochastic
volatility models in the context of FX derivative pricing are
increasingly popular. An example is the Wishart-based approach
proposed by \cite{BrangerMu} that focuses on the pricing of quanto
options. The Heston dynamics leads to an affine model which is
known to retain analytical tractability. We will provide a
complete discussion concerning the set of risk neutral measures
and the relation among model parameters under different
probability measures. Rather remarkably, as a consequence of the
specific Heston-type dynamics, the model remains functionally
invariant, after parameter rescaling, when the risk-neutral
measure is changed. This is a key feature that allows obtaining a
calibration with reasonable computational effort. We will then
test the model on real market data and show how a joint
calibration of the volatility smiles of EUR/USD/JPY and
AUD/USD/JPY triangles is possible. In- and out-of-sample
calibration tests will be reported to comment on the robustness of
the parameter estimation.

Previous analyses of the multi-dimensional FX volatility smile
problem have used different approaches to recover the risk neutral
probability distribution of the cross exchange rate, either by
means of joint densities or copulas, see \cite{austing2011},
\cite{bennettKennedy2004}, \cite{salmonSchneider2006}, and
\cite{hurdSalmonSchleicher2005}. Such contributions may be seen as
a generalization to the multi-dimensional setting of the classical
idea of \cite{bre78}, see also \cite{blissPani2002} and the
Gram-Charlier based approach in \cite{sch12}. The shortcoming of
these techniques is that they provide only a distribution for the
cross rate and not an explicit specification of the dynamics. In
the context of stochastic volatility models, \cite{CarrVerma}
propose a model with a single joint stochastic factor, which
however limits the flexibility to achieve satisfactory joint
calibrations. Another approach, in the presence of a SABR
stochastic volatility specification, is studied in
\cite{shirayaTakahashi2012} where asymptotic formulae are
presented.

The approach in this paper is fundamentally different.
Instead of putting the currency pairs at the basis of our model, we start from the
observation that any exchange rate may be seen as a ratio between
two quantities, the value of the currencies with respect to some
universal num\'eraire, and include this feature in the
specification of the model. \cite{flehu00} introduced the idea of a "natural numeraire", the value of which can be
expressed in different currencies, thus leading to consistent expressions for the FX rates as ratios. This point of view was also followed in \cite{article_heplaten06} and
\cite{bookplaten06} under the Benchmark approach. In this way, our model does not change
qualitatively depending on which perspective is used and there is no intrinsic difference
between main and cross currency pairs.  Independently of our work, the recent article by
\cite{doust2012} provides a stochastic volatility model of SABR
type where triangular relations hold. The approach is based on the concept of
intrinsic currency, introduced in \cite{article_Doust}.

Possible applications of the FX model we propose
are the valuation and risk management of multi-dimensional FX derivative options
and the possibility to reconstruct/simulate time series of less liquid cross currency pairs
from liquid ones, see \cite{doust2012}.

The paper is organized as follows: we present the model in Sec.\
2, initially using the perspective given by some kind of universal
num\'eraire. We continue with the basic properties of the model,
such as the presence of stochastic skewness, before presenting the
invariance of the model and transformation rule of its parameters
when the risk neutral measure is changed in Sec.\ 3. The explicit
formulae for the characteristic function and option prices are
given in Sec.\ 4. 
Finally, the joint calibration to EUR/USD/JPY and AUD/USD/JPY
market volatility smiles is presented in Sec.\ 5, together with a
discussion of the procedure and the results, including the Feller
condition and moment explosion.

\section{A Multifactor Heston-based exchange model}

We consider a foreign exchange market in which $N$ currencies are
traded between each other via standard FX spot and FX vanilla
option transactions. Inspired by the work of
\cite{article_heplaten06}, we start by considering the value of
each of these currencies in units of an artificial currency that
can be viewed as a universal num\'eraire. We will see that the
discussion is independent on the exact specification of this
num\'eraire. Let us work in the risk neutral measure defined by
the artificial currency and call $S^{0,i}(t)$ the value at time
$t$ of one unit of the currency $i$ in terms of our artificial
currency (so that $S^{0,i}$ can itself be thought as an exchange
rate, between the artificial currency and the currency $i$). We
model each of the $S^{0,i}$ via a multi-variate \cite{Heston93}
stochastic volatility model  with $d$ independent
Cox-Ingersoll-Ross (CIR) components \citep{cox85},
$\mathbf{V}(t)\in\mathbb R^d$. The dimension $d$ can be chosen
according to the specific problem and may reflect a PCA-type
analysis. We further assume that these stochastic volatility
components are \textit{common} between the different $S^{0,i}$.
Formally, we write
\begin{align}
    \frac{dS^{0,i}(t)}{S^{0,i}(t)}  &= (r^0-r^i) dt-(\mathbf{a}^i)^\top
                \sqrt{\mathrm{Diag}(\mathbf{V}(t))}d\mathbf{Z}(t),\quad \ i=1,\dots,N; \label{rates}\\
    dV_k(t)                        &= \kappa_k(\theta_k -V_k(t))dt + \xi_k \sqrt{V_k(t)}dW_k(t),\quad \ k=1,\dots,d;
\end{align}
where $\kappa_k, \theta_k, \xi_k\in \mathbb R$ are standard parameters in a CIR dynamics.
$\sqrt{\mathrm{Diag}(\mathbf{V})}$ denotes the diagonal matrix with the square root of the elements of the
vector $\mathbf{V}$ in the principal diagonal, this term is multiplied with the linear vector
$\mathbf{a}^i\in\mathbb R^d$ ($i=1,\dots,N$); as a result, the dynamics of the exchange rate is
driven by a linear projection of the variance factor $\mathbf{V}$
along a direction parametrized by $\mathbf{a}^i$, namely the total instantaneous variance is
$(\mathbf{a}^i)^\top \mathrm{\mathrm{Diag}}(\mathbf{V}(t)) \mathbf{a}^i dt$. In each monetary area $i$, the
money-market account accrues interest based on the deterministic risk free rate $r^i$,
\begin{align}
    dB^i(t)=&r^iB^i(t)dt,\quad \ i=1,\dots,N;
\end{align}
in our universal num\'eraire analogy $r^0$ is the artificial
currency rate. Finally, we assume that there is (only) a correlation between the innovations to $V_k$ and the innovations in the price with volatility $V_k$:
\begin{align}
    d\langle Z_k, W_h \rangle_t =& \rho_k \delta_{kh}dt, \quad k,h=1,\dots,d,\label{corr_struct}
\end{align}
 together with
$d\langle Z_k , Z_h \rangle_t = \delta_{kh}dt$ and $d\langle W_k , W_h \rangle_t = \delta_{kh}dt$, where
\begin{align*}
    \delta_{kh} & =\left\{
        \begin{array}{cc}
        1 & k=h,\\
        0 & k\not = h.
        \end{array}
        \right.
\end{align*}
This concludes the description of our model.\\

The idea behind this approach is that each exchange rate  is
driven by several independent drivers $Z_k$ ($k=1,..,d$), each
with an independent stochastic variance factor $V_k$, to which
$Z_k$ is partially correlated via $\rho_k$. The vectors
$\mathbf{a}^i$ ($i=1,\dots,N$) describe by how much each of the
different volatilities contributes to the dynamics of $S^{0,i}$.
This correlation structure is responsible for the appearance of
non-standard effects in the model,
like a stochastic skewness, as we will show in the sequel.\\

All in all, we have introduced a total number of parameters equal to $N_\mathrm{P} = Nd + 5d$
($Nd$ from the vectors $\mathbf{a}^i$ and 5 for each CIR process,  $\kappa_k, \theta_k, \xi_k,
\rho_k$ and the initial value $V_k(0)$) to describe the volatility skew of $(N^2-N)/2$ currency pairs.
As rule of thumb, assuming that each currency pair can be approximately modelled by a standard
one-dimensional Heston model, which is described by 5 parameters, around $5(N^2-N)/2$
parameters are needed to fit all volatility surfaces; the value of $d$ should be chosen to produce
approximately this number of parameters, if not less, to avoid instabilities due to overfitting.\\

Let us now turn our attention to the exchange rate $S^{i,j}$ between two
different currencies, say $i$ and $j$. We set by definition  $S^{i,j} = S^{0,j}/S^{0,i}$.
By straightforward calculation, we obtain for $i,j=1,..,N$:
\begin{align}
    \frac{dS^{i,j}(t)}{S^{i,j}(t)}    &= (r^i-r^j)dt+(\mathbf{a}^i -\mathbf{a}^j)^\top\mathrm{\mathrm{Diag}}(\mathbf{V}(t))\mathbf{a}^i dt
                                    +(\mathbf{a}^i - \mathbf{a}^j)^\top\sqrt{\mathrm{Diag}(\mathbf{V}(t))}d\mathbf{Z}(t). \label{fxrates}
\end{align}
At this stage we are still working under the risk neutral measure defined by the universal num\'eraire. The additional
drift term in (\ref{fxrates}) can be understood as a quanto adjustment between the artificial currency 0 and $i$.
Note also that the model is functionally symmetric with respect to which FX pairs we choose
to be the main ones and which one the cross\footnote{Note that if we start with a different specification
of the model it may happen that taking the dynamics of the ratio of a pair breaks the structure of the dynamics,
namely the model would not be functionally symmetric. This is what happens for example in the SABR model as
illustrated in the introduction, or if we consider a correlated GARCH volatility model, see e.g. the model
specifications in \cite{lewis2000}
and \cite{hullWhite87}.}. \\

Let us now  analyze some additional properties of the model and
familiarize with the meaning of the different parameters, starting from $\mathbf{a}^i$.
 A rather natural choice would be to set $\mathbf{a}^i$ equal to the canonical basis
$\mathbf{e}^i$ (i.e. the $i$-th element of the canonical basis of $\mathbb R^N$, $e^i_l=\delta_{li}, \ i,l=1,..,N$),
then (\ref{fxrates}), for $i\neq j$, reads (note that equal indices are not summed)
\begin{align}
    \frac{dS^{i,j}(t)}{S^{i,j}(t)}&=(r^i-r^j)dt+V^i(t) dt+\sqrt{V_i(t)}dZ^i(t)-\sqrt{V_ j(t)}dZ^j(t),
\end{align}
which in the 3-currency case leads to the 3-factor Heston model.  The problem with this choice is that the covariances
(and thus the correlations) between different pairs are forced to be positive,
\begin{align}
d \left\langle S^{i,j},S^{i,l}\right\rangle_t = S^{i,j}(t)S^{i,l}(t)V^i(t)dt \geq 0.
\end{align}
There is no empirical evidence for this inequality to hold in general between FX rates, see \cite{carrwu07}.
The additional vectors $\mathbf{a}^i$ are needed to describe a multi-dimensional FX market where the correlation may change sign.

To shed some additional light on the meaning of the vectors $\mathbf{a}^i$ we calculate the
  infinitesimal correlation $\varsigma^{i,j}$  between the log
returns of $S^{i,j}$ and the squared volatility $\mathrm{Vol}^2(S^{i,j})$: the quantity $\varsigma^{i,j}$ is known
to be related to the skewness of the distribution of the log-returns of the spot, see e.g. \cite{carrwu07}.
\begin{align}
    \varsigma^{i,j} (t) = \frac{d\langle \ln S^{i,j}, \mathrm{Vol}^2(S^{i,j}) \rangle_t}
                  {\sqrt{d\langle \ln S^{i,j}\rangle_t} \sqrt{d\langle \mathrm{Vol}^2(S^{i,j})\rangle_t}}.
\end{align}
Differently from standard single factor models, multifactor Heston models produce  stochastic skewness.
In fact, by means of straightforward calculations we obtain
\begin{align}
 d\left\langle \ln S^{i,j},\mathrm{Vol}^2(S^{i,j}) \right\rangle_t\nonumber& =d\left\langle\int_0^\cdot\sum_{l=1}^{d}{\left(a^i_l-a^j_l\right)
 \sqrt{V_l}dZ_{l}},\int_0^\cdot\sum_{k=1}^{d}\left(a^i_k-a^j_k\right)^2\xi_k\sqrt{V_k}dW_{k} \right\rangle_t\nonumber\\
&=\sum_{k=1}^{d}\left(a^i_k-a^j_k\right)^3\xi_kV_k(t)\rho_kdt\qquad i,j=1,..,N.
\end{align}
Combining this term with
\begin{align*}
d    \left\langle \ln S^{i,j},\ln S^{i,j}\right\rangle_t
    &=
    \left(\mathbf{a}^i-\mathbf{a}^j\right)^\top\mathrm{Diag}(\mathbf{V}(t))\left(\mathbf{a}^i-\mathbf{a}^j\right)dt\\
&= \sum_{k=1}^{d}\left(a^i_k-a^j_k\right)^2V_k(t)dt\end{align*}
and
\begin{equation*}
d\left\langle \mathrm{Vol}^2(S^{i,j}), \mathrm{Vol}^2(S^{i,j})\right\rangle_t=\sum_{k=1}^{d}\left(a^i_k-a^j_k\right)^4\xi_k^2V_k(t)dt
\end{equation*}
gives
\begin{align}
\varsigma^{i,j} (t)=\frac{\sum_{k=1}^{d}\left(a^i_k-a^j_k\right)^3\xi_kV_k(t)\rho_k}{\sqrt{\sum_{k=1}^{d}\left(a^i_k-a^j_k\right)^4\xi_k^2V_k(t)}\sqrt{\sum_{k=1}^{d}\left(a^i_k-a^j_k\right)^2V_k(t)}}.
\end{align}

  The
vectors $\mathbf{a}^i$ are therefore directly related to the amount of skewness for each of the different exchange
rates.
This quantity is stochastic due to the presence of the variance factors $V_k$ (see also \cite{ChristHest} who found similar results in the equity market using a multi-Heston framework) and can assume positive as negative sign according to the relative importance of the coefficients $a^i_k,a^j_k$ in the summation, that is according to the relative importance of the volatility factor $V_k$ in each currency. The same argument applies to the instantaneous covariance between the assets that can be written as\footnote{Note that \begin{align*}
d\langle S^{ij}, S^{il}\rangle_t&= (X^{i,j})^\top X^{i,l} dt,
\end{align*}
where
 \begin{align*}
X^{i,j} &= S^{ij}(t)\sqrt{\mathrm{Diag}({\bf V}(t))} (\mathbf{a}^i-\mathbf{a}^j)\\
X^{i,l} &= S^{il}(t)\sqrt{\mathrm{Diag}({\bf V}(t))}
(\mathbf{a}^i-\mathbf{a}^l).
\end{align*}
Since $\mathbf{a}^i,\mathbf{a}^j,\mathbf{a}^l$ are arbitrary real vectors, the infinitesimal covariance can be any real number,
and the corresponding infinitesimal correlation spans the entire interval [-1,1].
}
\begin{align*}
d\langle S^{ij}, S^{il}\rangle_t&= S^{ij}(t)S^{il}(t)(a^i-a^j)^\top \mathrm{Diag}({\bf V}(t)) (a^i-a^l)dt\\
                     &= S^{ij}(t)S^{il}(t)\sum_{k=1}^d (a_k^i-a_k^j)(a_k^i-a_k^l)V_k(t) dt.
\end{align*}

\section{Num\'eraire invariance}

Up to now we have worked under the risk neutral measure defined by our (rather unspecified) artificial currency.
In practical pricing applications, it is more convenient to change the num\'eraire to any of the currencies included in our
FX multi-dimensional system. Without loss of generality, let us consider the risk neutral measure defined by the $i$-th
money market account $B^i$ and derive the dynamical equations for the standard FX rate $S^{i,j}$.\\

Under the assumptions of the fundamental theorem of asset pricing (cfr. e.g. \cite{bjork}, chapters 13 and 14), investing into the foreign money market account gives a traded asset with value $S^{i,j}B^j$, whose discounted value has to be a $\mathbb{Q}^i$-martingale.
Hence,
\begin{align}
   d\left(\displaystyle \frac{S^{i,j}(t)B^j(t)}{B^i(t)}\right)  & = \frac{S^{i,j}(t)B^j(t)}{B^i(t)}
                     \left((\mathbf{a}^i - \mathbf{a}^j)^\top\mathrm{Diag}(\mathbf{V}(t))
                       \mathbf{a}^i dt+(\mathbf{a}^i - \mathbf{a}^j)^\top
                   \sqrt{\mathrm{Diag}(\mathbf{V}(t))}d\mathbf{Z}(t)\right)\nonumber\\
                     & = \frac{S^{i,j}(t)B^j(t)}{B^i(t)}(\mathbf{a}^i - \mathbf{a}^j)^\top\sqrt{\mathrm{Diag}(\mathbf{V}(t))}
                    d\mathbf{Z}^{\mathbb{Q}^i}(t).
\end{align}
In the last line we implicitly defined the new Brownian motion vector
$\mathbf{Z}^{\mathbb{Q}^i}$ under the measure $\mathbb{Q}^i$ by imposing the
$\mathbb{Q}^i$-local martingale property and by Girsanov theorem
\begin{align}
    d\mathbf{Z}(t)^{\mathbb{Q}^i}=d\mathbf{Z}(t)+\sqrt{\mathrm{Diag}(\mathbf{V}(t))}\mathbf{a}^i dt,\quad i=1,..,N.
\label{girsanovQi}\end{align}
%
%
The $\mathbb Q^i$ risk neutral dynamics of the exchange rate $S^{i,j}$ becomes
\begin{align} \label{fxrateij_qi}
    dS^{i,j}(t) & = S^{i,j}(t)\left((r^i-r^j)dt+(\mathbf{a}^i  -
    \mathbf{a}^j)^\top\sqrt{\mathrm{Diag}(\mathbf{V}(t))}d\mathbf{Z}^{\mathbb{Q}^{i}}(t)\right),
\end{align}
as desired.

The measure change has also an impact on the variance processes, via the correlations $\rho_k, k=1,..,d$,
\begin{align}
    dW_k^{\mathbb{Q}^i}(t) & = dW_k(t) + \rho_k \left(\mathbf{e}^k\right)^\top \sqrt{\mathrm{Diag}(\mathbf{V}(t))}\mathbf{a}^i dt.
\end{align}
%
We finally obtain the dynamic equations under the new measure. With an appropriate redefinition of the CIR parameters
\begin{align}
    \rho_k^{\mathbb{Q}^i}    = & \rho_k, \nonumber \\
    \kappa_k^{\mathbb{Q}^i}  = & \kappa_{k}+\xi_k\rho_k a^i_k, \nonumber \\
    \theta_k^{\mathbb{Q}^i}  = & \theta_k\frac{\kappa_k}{\kappa_k^{\mathbb{Q}^i}},\nonumber
\end{align}
we can recast the variance SDE in its original form
\begin{equation}
    dV_{k}(t)= \kappa_k^{\mathbb{Q}^i}(\theta_k^{\mathbb{Q}^i} - V_k(t)) dt + \xi_k \sqrt{V_k(t)}dW_k^{\mathbb{Q}^i}(t).\label{VunderQi}
\end{equation}
Financially, it makes sense to
enforce mean reversion of the variance, other than mean explosion, yielding a condition on $\kappa_k^{\mathbb{Q}^i} > 0$
or conversely on the original model parameters $\kappa_{k}$, $\xi_k$, $\rho_k$, and
$a^i_k$.\\
By applying
Girsanov theorem again, this time switching to the $\mathbb{Q}^j$ risk neutral measure, the CIR parameters become
\begin{align}
    \kappa_k^{\mathbb{Q}^j} & = \kappa_k^{\mathbb{Q}^i} +\rho_k \xi_k (a^j_k-a^i_k),\nonumber\\
    \theta_k^{\mathbb{Q}^j} & = \theta_k^{\mathbb{Q}^i}\frac{\kappa_k^{\mathbb{Q}^i}}{\kappa_k^{\mathbb{Q}^j}},\label{cirparams}
\end{align}
together with the invariant $\rho_k^{\mathbb{Q}^j}  = \rho_k^{\mathbb{Q}^i}$ and $\xi_k^{\mathbb{Q}^j}  =
\xi_k^{\mathbb{Q}^i}$. These are the fundamental transformation rules for the model parameters.
The invariance of the functional form of the model under measure change is an appealing feature of our model; other
specifications of the stochastic volatility might break this symmetry.

\section{Option pricing}

Together with the symmetry of the model specification with respect to the num\'eraire choice, a second central
feature of the model is the availability of a (semi)-analytical solution for {\it all} vanilla option prices.
The pricing formula itself is symmetric with respect to the choice of the option underlying, once we work under the risk neutral measure
associated with one of the currencies involved in the option and the parameters are transformed via (\ref{cirparams}).\\

Let us consider a call option $C(S^{i,j}(t),K^{i,j},\tau), i,j=1,..,N,i\not=j,$ on a generic FX rate $S^{i,j}(t) = \exp(x^{i,j}(t))$ with
strike $K^{i,j}$, maturity $T$ ($\tau = T - t$ is the time to maturity)  and face equal to one unit of the foreign currency.
We write for the CIR parameters $\kappa_k = \kappa_k^{\mathbb{Q}^i}, \theta_k=\theta_k^{\mathbb{Q}^i}$ and so on, implicitly assuming that
they have  been transformed via (\ref{cirparams}) in the $i$-th risk neutral measure ${\mathbb{Q}^i}$. Being an affine model,
the (generalized) characteristic function conditioned on the initial values
\begin{align}
    \phi^{i,j}(\omega,t, \tau, x ,\mathbf{V}) =  \mathbb{E}^{\mathbb{Q}^i}_t[e^{\mathtt{i}\omega x^{i,j}(T)}|x^{i,j}(t) = x,\mathbf{V}(t) = \mathbf{V}  ]
\end{align}
can be derived analytically (here $\mathtt{i}=\sqrt{-1}$).  Standard numerical integration methods can then be used to invert
the Fourier transform to obtain the probability density at $T$ or the vanilla price via integration
against the payoff, with overall limited computational effort.
By applying standard arguments (see e.g. \cite{lewis2000}, \cite{article_Lipton}, \cite{sepp03}) the value of a call option
can be expressed in terms of the integral of the product of the Fourier transform of the payoff and the generalized characteristic
function of the log-asset price\footnote{Here we adopt the pricing method of \cite{lewis2000} who uses the characteristic function computed with a complex argument, also called generalized characteristic function. The complex argument $\omega$ belongs to a strip of regularity for  the function $\phi^{i,j}$ in order to be able to integrate the payoff function. On the other hand, this method generalizes the methodology introduced by \cite{article_Carr99} which involves the introduction of the so-called damping integrating factor.}:
\begin{align}
    C(S^{i,j}(t),K^{i,j},\tau)& =e^{-r^i\tau}\frac{1}{ 2\pi }\int_{\mathcal{Z}}\phi^{i,j}(-\lambda,t, \tau, x ,\mathbf{V}){\Phi}(\lambda)d\lambda,  \label{price}
\end{align}%
where
\begin{equation*}
    \Phi(\lambda)=\int_{\mathcal{Z}}e^{\mathtt{i} \lambda x}\left( e^{x}-K^{i,j}\right)^+dx
\end{equation*}
is the Fourier transform of the payoff function and $\mathcal{Z}$
denotes the strip of regularity of the payoff, that is the
admissible domain where the integral in (\ref{price}) is well
defined. In other words, the pricing problem is essentially solved
once the (conditional) characteristic function of the log-exchange
rate is known. In what follows we calculate the moment generating
function $G^{i,j}(\omega,t, \tau, x ,\mathbf{V})$ (Laplace
tranform) from which the characteristic function is easily derived
via a rotation in the complex plane $\phi^{i,j}(\omega,t, \tau, x
,\mathbf{V})=G^{i,j}(\mathtt{i}\omega,t, \tau, x ,\mathbf{V})$.
The conditional Laplace transform is of a particularly simple form, which is exponentially affine in the initial state of the process
\begin{align}
    G^{i,j}(\omega,t, \tau, x ,\mathbf{V})&=\exp\left[\omega x+\left(r^i-r^j\right)\omega(\tau)+\sum_{k=1}^{d}{\left(A_{k}^{i,j}(\tau)+B^{i,j}_{k}(\tau)V_k\right)}\right],
\end{align}\label{LTMH1}
where for $k=1,..,d$:
\begin{align}
    A^{i,j}_k(\tau)&=\frac{2\kappa_k\theta_k}{\xi_k^{2}}\log{\frac{\lambda_k^+-\lambda_k^-}{\lambda_k^+e^{\lambda_k^-(\tau)}-\lambda_k^-e^{\lambda_k^+(\tau)}}};\label{LTMH3}\\
    B^{i,j}_k(\tau)&=\frac{\left(\omega^2-\omega\right)}{2}\left(a^i_k-a^j_k\right)^2\frac{1-e^{-\sqrt{\Delta_k}\tau}}{\lambda_k^+e^{-\sqrt{\Delta_k}\tau}-\lambda_k^-};\label{LTMH2}\\
    \Delta_k&=\left(-\kappa_k+\omega\left(a^i_k-a^j_k\right)\rho_k\xi_k\right)^2-\xi_k^2\left(\omega^2-\omega\right)\left(a^i_k-a^j_k\right)^2;\\
    \lambda_k^\pm&=\frac{\left(-\kappa_k+\omega\left(a^i_k-a^j_k\right)\rho_k\xi_k\right)\pm\sqrt{\Delta_k}}{2}.
\end{align}
The derivation of this formula can be found in the Appendix.


\section{Simultaneous calibration of FX triangles}\label{sec:calib}

\subsection{Setup}

In this section we show an example of simultaneous calibration to
three market volatility surfaces of options. We consider two
currency triangles: EUR/USD/JPY as it appeared on the day
23/7/2010 and AUD/USD/JPY on 2/11/2012. Differences between the
two sets of market data are due both to the different currency
pair involved, e.g. replacing the EUR with the larger yield
carry-trade AUD currency, and the different time-stamp. Note in
particular the pronounced skew in the EURUSD volatility (bid for
the EURUSD puts, see Fig.\ \ref{fig:USDEUR}) during the European
debt crisis in mid 2010 and the almost symmetric shape of the
USDJPY volatilities, usually bid for the USDJPY puts, prior to BOJ
currency easing efforts and the then ongoing USD rally in late
2012, see Fig.\ \ref{fig:USDJPY6}. For each of the two case
studies we consider the implied volatility surfaces for each pair
in the triangle, eg. for EUR/USD/JPY, the pairs USD/EUR, USD/JPY
and EUR/JPY, that is $N=3$ with $i = \mathrm{USD};\ \mathrm{EUR};\
\mathrm{JPY}$. The volatility sample includes expiry dates ranging
from 3 days to 5 years. The quotes follow the standard Delta
quoting conversion in the FX option market, we have quotes on DN,
25 Delta, 15 Delta, and 10 Delta\footnote{\label{mktconv}It is important to
stress that in the forex market implied volatilities surfaces are
expressed in terms of maturity and Delta (see e.g.
\cite{Wystup10}, \cite{bookClark11}): the market practice is to
quote volatilities for strangles and risk reversals which can then
be employed to reconstruct a whole surface of implied volatilities
via an interpolation method (see e.g. \cite{Wystup10},
\cite{bookWystup06}, \cite{bookClark11}). Once we have the quotes
in terms of Delta, to perform the calibration we have to convert
Deltas into strike prices. The procedure can be found
e.g. in \cite{BenHuiz03}.}.\\

Let us concentrate here on the EUR/USD/JPY example, the other case follows the same procedure.
We try to fit simultaneously the three volatility surfaces using two stochastic
drivers, $d = 2$. This choice yields a total number of parameters $N_P = 16$,
comparable to the number of parameters in 3 independent Heston models (15
parameters). This choice should not lead to overfitting instabilities. We work under
the USD risk neutral measure to derive the option prices of the pairs EUR/USD
and USD/JPY and the EUR measure for the EUR/JPY options, using (\ref{price}). We
calibrate the CIR parameters in the USD measure $\kappa_k^{\scriptscriptstyle\mathrm{USD}},
\theta_k^{\scriptscriptstyle\mathrm{USD}},\xi_k^{\scriptscriptstyle\mathrm{USD}},\rho_k^{\scriptscriptstyle\mathrm{USD}},\ k=1,2$. The
parameters for the EUR/JPY are transformed to the EUR measure through
Eqs.\ (\ref{cirparams}) and the invariance property of correlation and vol-of-vol parameters.\\

The calibration is done via a standard non-linear least-squares optimizer that
minimizes the total calibration error in terms of the difference between calibrated
and target implied volatities $\sum_{n}(\sigma_{n,\mathrm{market}}^\mathrm{imp}-\sigma_{n,\mathrm{model}}^\mathrm{imp})^2$.
The use of a norm in price should be avoided as the numerical range for option prices may be large,
thus introducing a bias in the optimization. In fact, a norm in price penalizes greatly high prices,
so that the fit for short maturity options (which are cheaper) is quite poor. For a more detailed discussion
on the impact of the penalizing function on the calibrated parameters we refer to \cite{ChristJacobs04} and \cite{dafgra11}.

\subsection{Calibration results}
In Figs.\ \ref{fig:USDEUR}, \ref{fig:USDJPY} and \ref{fig:EURJPY}
we plot the market implied volatilities against those produced by
the model. The plots refer to the largest sample in Table\
\ref{tab:params}. Market volatilities are denoted by crosses,
model volatilities are denoted by circles. The quality of the fit
is comparable if not superior with respect to what is usually
achieved by means of the standard Heston model for a single
currency pair. The plots for the calibration on the sub-samples
are completely analogous.

In Table \ref{tab:params} we report the result of the calibration of the model for the EUR/USD/JPY triangle, whereas Fig.\ \ref{fig:SE_1},
reports the squared error in volatility for each moneyness/maturity. We have performed the
optimization considering different sets of expiries. The expiries considered in the largest sample are the following:
1, 2, 3, 6, 9 months and 1 year. The result for this particular choice of expiries is reported in the first column on the left.
Then we have repeated the experiment by excluding the largest expiry, 1 year. The result is reported in the second column. We
proceed in this way by excluding more and more expiries. The smallest sample is reported in the last column and considers
only options expiring in 1 and 2 months. In-sample squared errors in implied volatilities are visualized in Figs.\ \ref{fig:SE_2},
\ref{fig:SE_3}, \ref{fig:SE_4} \ref{fig:SE_5}.

For the AUD/USD/JPY case, we limit ourselves to report in Figs.\ \ref{fig:USDJPY6},
\ref{fig:AUDUSD6} and \ref{fig:AUDJPY6} the result of the fit on the largest sample, consisting of implied
volatilities at 1, 2, 3, 6, 9 months and 1 year. In Fig.\ \ref{fig:SE_11} we show the squared errors in implied
volatilities for each point of the surface that we are considering. We report in Table\ \ref{tab:params_AUDUSDJPY} the
calibrated parameters. Also in this case the calibration yields
a satisfactory fit to the market data.

\subsection{Parameters stability tests}

In this subsection we comment on the stability of the parameters via two different
types of analysis. We first measure the impact on the parameters resulting from the calibration procedure.
Secondly, we fit the model parameters to a certain sample and then use these parameters to price an option
which is not included in the sample. If the out-of-sample prices are close to the market, the model gives
a reasonable description of the joint underlying FX rates dynamics. Moreover, the calibration can be done
on a limited set of expiries, reducing the computation effort of the optimizer.\\

As far as the first analysis is concerned, we show in Table\ \ref{tab:variation} the relative variations
computed with respect to the largest sample. With the exception of $\kappa_1,\kappa_2$ we can see that there is a
good degree of stability of the parameters across the sub-samples. Consequently, we perform also a second calibration experiment, where we fix $\kappa_k=1,k=1,2$.
The results of this experiment are outlined in Table\ \ref{tab:params2}. The relative variation of the parameters
can be found in Table\ \ref{tab:variation2}. We notice that with this choice we get a good degree of stability,
the most relevant fluctuation is now around 20\% for $\theta_1$\footnote{We do
not report, for the sake of brevity, the volatility surfaces arising from this last
experiment, but the quality of the fit is the same as before.}.  \\

Let us now turn our attention to the out-of-sample exercise. In Tables\ \ref{tab:OOS1}, \ref{tab:OOS2}, \ref{tab:OOS3} and
\ref{tab:OOS4} we show the difference between the market and the out-of-sample volatility for all sub-samples. The differences
are always well below one volatility point.

\subsection{Moment explosion}

In this subsection we discuss some caveats coming from the parameters we obtained through the calibration procedure.
It is known that for square root processes $0$ represents
an attainable state when the Feller condition is not satisfied, that is when $2\kappa_k\theta_k<\xi_k^2$.
In our modelling framework we have two volatility factors, hence we can perform the check for each factor. In
Table\ \ref{tab:Feller} we report the quantity $F_k=2\kappa_k\theta_k-\xi_k^2, k=1,2$. We observe violations of the
Feller condition, which constitutes a well-known fact in the FX derivative practice, shared with
the standard one-dimensional Heston model, see\ \cite{bookClark11}. This phenomenon is strictly related to another
established fact in stochastic volatility models, namely the pathological moment explosions which might often
impact the stability of the pricing tools, see e.g.\ \cite{article_AndPit}, \cite{article_KRME} and \cite{article_GlassKim}.
The model dynamics might lead to the explosion of moments, which become infinite in finite time. This fact might
lead to complications/instabilities in standard numerical pricing routines mostly for large maturities.
We can calculate the time of moment explosion for all currency pairs, see \cite{article_AndPit}. In Table\ \ref{tab:expl}
we consider moments up to order 5.\\

\section{Conclusions}

We have introduced a new multi-factor stochastic volatility Heston-based model that can
provide an accurate joint description of multiple FX vanilla options across different currency pairs. The emphasis in the model specification
has been in the preservation of the specific symmetries of FX markets. Differently from other asset classes, appropriate
multiplications/divisions and inversions of FX rates are still FX rates.
The choice of our simple CIR-based dynamics for the stochastic variance is instrumental in achieving this symmetry. 
We have indeed proven that our model is invariant with respect to the choice of the num\'eraire once the model parameters
are appropriately transformed. The model is always of affine-type independently of which currency is used as risk free, leading
to semi-analytical expression for all vanilla options between {\it any} of two currencies. This property is crucial when
it comes to calibrating the model. In a standard global optimization algorithm we can consider together vanilla options in
all currency pairs and achieve a simultaneous fit to the different volatility surfaces with reasonable computational effort.

The model shares naturally several stylized facts with the Heston
model. The Feller condition is often violated when fitting the
model to FX volatility surfaces, a common observation in the
practice. Moreover, higher moments of the spot distribution
explode at finite time; a property that might lead to
complications/instabilities in standard numerical pricing routines
mostly when maturities are large. Finally, like any pure
stochastic volatility model, our model cannot be expected to
deliver a perfect calibration of the vanilla surfaces across all
Deltas and tenors, especially in the short end.

Having said that, the main result of the paper is a promising joint calibration of the model to the implied volatilities smiles of the EUR/USD/JPY
and AUD/JPY/USD FX triangles. The fit remains satisfactory across the currency pairs, Deltas and tenors which were considered.
Several in- and out-of-sample calibration studies in fact have proven the robustness of the calibration, especially once the mean reversion speed
$\kappa$ has been fixed. Asymptotic expansions of the implied volatility surface are also included in the Appendix as they shed light
on the meaning of the different model parameters and can help speeding up the calibration procedure by giving an educated guess
for the initial parameters in the optimization procedure.

The price to pay in order to obtain a consistent simultaneous calibration to all volatilities surfaces is that the instantaneous volatilities of the
currency pairs do not have single dedicated drivers. Their dynamics is rather brought about by a linear combination of several hidden stochastic
factors. As in any principal component analysis, it is not easy to assign a financial meaning to each model parameter. As this study
has shown, this appealing feature has most likely to be traded away in order to capture the complex phenomenology of the present global and widely
interconnected FX markets.

\section{Acknowledgements}

We are grateful to Jan Baldeaux, Damiano Brigo, Imran Hafeez,
Patrick Kuppinger, Eckhard Platen, Wolfgang Runggaldier and an
anonymous referee for useful comments.

\newpage

\section{Images and Tables}
\subsection{Calibration of EUR/USD/JPY}
\begin{table}[H]
    \centering
        \begin{tabular}{ccccccc}
&6&5&4&3&2\\
\hline\\
$V_1$& 0.0137  &  0.0137 &   0.0136 &   0.0137 &   0.0135\\
$V_2$& 0.0391  &  0.0365 &   0.0278 &   0.0293 &   0.0273\\
$a^{\scriptscriptstyle\mathrm{USD}}_1$& 0.6650  &  0.6713 &   0.6165 &   0.6371 &   0.6518\\
$a^{\scriptscriptstyle\mathrm{USD}}_2$& 1.0985  &  1.0531 &   0.9700 &   0.9795 &   0.9514\\
$a^{\scriptscriptstyle\mathrm{EUR}}_1$& 1.6177  &  1.6222 &   1.5648 &   1.5804 &   1.6061\\
$a^{\scriptscriptstyle\mathrm{EUR}}_2$& 1.3588  &  1.3208 &   1.2746 &   1.2797 &   1.2737\\
$a^{\scriptscriptstyle\mathrm{JPY}}_1$ & 0.2995  &  0.3151 &   0.2732 &   0.3035 &   0.3116\\
$a^{\scriptscriptstyle\mathrm{JPY}}_2$ & 1.6214  &  1.5922 &   1.5882 &   1.5858 &   1.5816\\
$\kappa_1$& 0.9418  &  1.1432 &   1.5138 &   1.7349 &   1.8685\\
$\kappa_2$& 1.7909  &  1.9998 &   1.9014 &   0.7142 &   0.7210\\
$\theta_1$& 0.0370  &  0.0349 &   0.0329 &   0.0329 &   0.0297\\
$\theta_2$& 0.0909  &  0.0839 &   0.0670 &   0.1236 &   0.1091\\
$\xi_1$& 0.4912  &  0.5138 &   0.5542 &   0.5847 &   0.5962\\
$\xi_2$& 1.0000  &  0.9997 &   0.8736 &   0.8318 &   0.8568\\
$\rho_1$& 0.5231 &   0.5118  &  0.4916   & 0.4727   & 0.4567\\
$\rho_2$& -0.3980&   -0.3956 &  -0.3943  & -0.3902  & -0.3728\\
Res. norm.& 4.6996e-004&3.4244e-004&1.8618e-004&1.1145e-004&5.2514e-005\\
\hline
        \end{tabular}
    \caption{This table reports the results of the calibration of the model.
    We concentrate on the two factor case. For each column, a different number of expiries,
    ranging from 6 to 2, is chosen. More specifically, 6 means that the following expiries are
    considered: 1, 2, 3, 6, 9 months and 1 year, whereas 5 means that the longest maturity, i.e. 1
    year is excluded from the sample. We proceed analogously in the subsequent columns by excluding
    the longest expiry date up to the point where we perform the calibration on the 2-sample, where
    we fit the smile at 1 and 2 months. We consider market data as of 23rd July 2010. The reference
    exchange rates are $S^{\scriptscriptstyle\mathrm{JPY},\mathrm{EUR}}(0)=112.29$,
    $S^{\scriptscriptstyle\mathrm{USD},\mathrm{EUR}}(0)=1.2921$ and $S^{\scriptscriptstyle\mathrm{JPY},\mathrm{USD}}(0)=86.90$.
    Res. norm. is the residual of
    the objective function for the given set of parameters.}
    \label{tab:params}
\end{table}

\begin{table}[H]
    \centering
        \begin{tabular}{ccccccc}
&5&4&3&2\\
\hline\\
$V_1$   &   0.1244\% &  -0.2866\% &    0.0960\% &   -1.1068\%\\
$V_2$&  -6.5645\% &  -28.9269\% &  -25.0960\% &  -30.0900\%\\
$a^{\scriptscriptstyle\mathrm{USD}}_1$& 0.9368\% &   -7.3035\% &   -4.1928\% &   -1.9883\%\\
$a^{\scriptscriptstyle\mathrm{USD}}_2$& -4.1309\% &  -11.6957\% &  -10.8309\% &  -13.3918\%\\
$a^{\scriptscriptstyle\mathrm{EUR}}_1$&  0.2745\% &   -3.2714\% &   -2.3082\% &   -0.7190\%\\
$a^{\scriptscriptstyle\mathrm{EUR}}_2$&  -2.7989\% &   -6.1962\% &   -5.8255\% &   -6.2652\%\\
$a^{\scriptscriptstyle\mathrm{JPY}}_1$ & 5.1809\% &   -8.8010\% &    1.3206\% &    4.0245\%\\
$a^{\scriptscriptstyle\mathrm{JPY}}_2$ &  -1.7985\% &   -2.0460\% &   -2.1910\% &   -2.4522\%\\
$\kappa_1$&  21.3845\% &   60.7349\% &   84.2055\% &   98.3943\%\\
$\kappa_2$& 11.6649\% &    6.1715\% &  -60.1213\% &  -59.7402\%\\
$\theta_1$& -5.6226\% &  -11.1784\% &  -11.0318\% &  -19.7810\%\\
$\theta_2$& -7.7145\% &  -26.3082\% &   36.0430\% &   20.0453\%\\
$\xi_1$& 4.6020\% &   12.8359\% &   19.0424\% &   21.3784\%\\
$\xi_2$&  -0.0244\%&  -12.6344\%&  -16.8201\%&  -14.3193\%\\
$\rho_1$& -2.1702\% &   -6.0305\% &   -9.6495\% &  -12.7003\%\\
$\rho_2$& -0.6031\% &   -0.9375\% &   -1.9522\% &   -6.3351\%\\
\hline
        \end{tabular}
    \caption{In this table we consider the calibration on the largest sample as a basic case.
    We report the percentage difference between the model parameters resulting from the subsamples.}
    \label{tab:variation}
\end{table}

\begin{table}[H]
    \centering
        \begin{tabular}{ccccccc}
&6&5&4&3&2\\
\hline\\
$V_1$&  0.0438  &  0.0430  &  0.0405  &  0.0421  &  0.0412\\
$V_2$&  0.0465  &  0.0450  &  0.0408  &  0.0370  &  0.0335\\
$a^{\scriptscriptstyle\mathrm{USD}}_1$&0.7201  &  0.7165 &   0.7086 &   0.7099 &   0.7082\\
$a^{\scriptscriptstyle\mathrm{USD}}_2$&1.0211 &   1.0182  &  1.0095 &   0.9915 &   0.9685\\
$a^{\scriptscriptstyle\mathrm{EUR}}_1$&1.2517 &   1.2534  &  1.2603 &   1.2477 &   1.2538\\
$a^{\scriptscriptstyle\mathrm{EUR}}_2$&1.2624 &   1.2616  &  1.2619 &   1.2575 &   1.2589\\
$a^{\scriptscriptstyle\mathrm{JPY}}_1$ &0.5159 &   0.5155  &  0.5093 &   0.5206 &   0.5142\\
$a^{\scriptscriptstyle\mathrm{JPY}}_2$ &1.5053 &   1.5083  &  1.5223 &   1.5307 &   1.5372\\
$\theta_1$&0.1154 &   0.1169  &  0.1203 &   0.1391 &   0.1300\\
$\theta_2$&0.1344 &   0.1377  &  0.1350 &   0.1253 &   0.1081\\
$\xi_1$&0.8892 &   0.8898  &  0.8992 &   0.9700 &   0.9925\\
$\xi_2$&0.9338 &   0.9450  &  0.9458 &   0.9616 &   0.9659\\
$\rho_1$& 0.5226  &  0.5132   & 0.4950  &  0.4756  &  0.4591\\
$\rho_2$& -0.4042 &  -0.4030  & -0.4004 &  -0.3887 &  -0.3721\\
Res. Norm.&  0.0013& 4.7824e-04& 1.9968e-04& 2.6412e-04&4.7716e-04\\
\hline
        \end{tabular}
    \caption{This table reports the results of the calibration of the model. In this case we are
    assuming $\kappa_k=1,k=1,2$. For each column, a different number of expiries, ranging from 6 to 2,
    is chosen. Res. norm. is the residual of the objective function for the given set of parameters.}
    \label{tab:params2}
\end{table}

\begin{table}[H]
    \centering
        \begin{tabular}{ccccccc}
&5&4&3&2\\
\hline\\
$V_1$   &-1.8915\%&   -7.6930\%&   -3.8971\%&   -5.9284\%\\
$V_2$&  -3.1003\%&  -12.1687\%&  -20.3324\%&  -28.0033\%\\
$a^{\scriptscriptstyle\mathrm{USD}}_1$&-0.5056\%&   -1.6003\%&   -1.4171\%&   -1.6497\%\\
$a^{\scriptscriptstyle\mathrm{USD}}_2$&-0.2832\%&   -1.1348\%&   -2.9033\%&   -5.1535\%\\
$a^{\scriptscriptstyle\mathrm{EUR}}_1$& 0.1322\%&    0.6825\%&   -0.3164\%&    0.1703\%\\
$a^{\scriptscriptstyle\mathrm{EUR}}_2$&-0.0691\%&   -0.0438\%&   -0.3903\%&   -0.2826\%\\
$a^{\scriptscriptstyle\mathrm{JPY}}_1$ &-0.0857\%&   -1.2718\%&    0.9087\%&   -0.3396\%\\
$a^{\scriptscriptstyle\mathrm{JPY}}_2$ & 0.1994\%&    1.1235\%&    1.6872\%&    2.1131\%\\
$\theta_1$& 1.3740\%&    4.2649\%&   20.5911\%&   12.6966\%\\
$\theta_2$& 2.4412\%&    0.4181\%&   -6.8073\%&  -19.5908\%\\
$\xi_1$& 0.0622\%&    1.1246\%&    9.0864\%&   11.6096\%\\
$\xi_2$& 1.1985\%&    1.2881\%&    2.9747\%&    3.4434\%\\
$\rho_1$& -1.7994\%&   -5.2886\%&   -8.9844\%&  -12.1501\%\\
$\rho_2$& -0.2922\%&   -0.9392\%&   -3.8297\%&   -7.9304\%\\
\hline
        \end{tabular}
    \caption{In this table we consider the calibration on the largest sample as a basic case,
    when $\kappa_k=1,k=1,2$. We report the percentage difference between the model parameters
    resulting from the subsamples.}
    \label{tab:variation2}
\end{table}

\begin{table}[H]
    \centering
        \begin{tabular}{cccc}
        &USD/EUR&USD/JPY&EUR/JPY\\
        \hline\\
10DC&-0.0006& -0.0002&-0.0027\\
15DC&       & 0.0003&-0.0017\\
25DC&-0.0012& 0.0005&-0.0005\\
0   &-0.0022& 0.0009&0.0021\\
25DP&-0.0008& 0.0012&0.0042\\
15DP&       & 0.0004&0.0031\\
10DP& 0.0009& -0.0001&0.0011\\
\hline
        \end{tabular}
    \caption{Out-of-sample performance. This table reports the raw difference between the market implied volatility and the model generated 
    implied volatility for 1 year, when we calibrate the model to the previous 5 expiries. Moneyness levels follow the standard Delta
quoting convention in the FX option market, see Footnote \ref{mktconv}. DC and DP stand for "delta call" and "delta put" respectively. Blanks
    on the first column reflect missing market data for 15DC and 15DP.}
    \label{tab:OOS1}
\end{table}

\begin{table}[H]
    \centering
        \begin{tabular}{ccccccc}
        &USD/EUR&USD/EUR&USD/JPY&USD/JPY&EUR/JPY&EUR/JPY\\
        &9m&1y&9m&1y&9m&1y\\
        \hline\\
10DC&-0.0031 &   0.0003 &  -0.0013 &  -0.0001 &  -0.0006 &  -0.0019\\
15DC&       &         & -0.0007  &  0.0002  &  0.0008  & -0.0012\\
25DC&-0.0023 &  -0.0007 &   0.0004 &   0.0004 &   0.0021 &  -0.0006\\
0   & -0.0021&   -0.0021&    0.0020&    0.0006&    0.0036&    0.0012\\
25DP& -0.0010&   -0.0007&    0.0011&    0.0010&    0.0028&    0.0033\\
15DP&       &         &   -0.0008&    0.0003&    0.0004&    0.0025\\
10DP&-0.0006 &   0.0015 &  -0.0014 &  -0.0000 &  -0.0026 &   0.0007\\
\hline
        \end{tabular}
    \caption{Out-of-sample performance. This table reports the raw difference between the market implied volatility and the model generated implied volatility for 1 year and 9 months, when we calibrate the model to the previous 4 expiries. Moneyness levels follow the standard Delta
quoting convention in the FX option market, see Footnote \ref{mktconv}. DC and DP stand for "delta call" and "delta put" respectively. Blanks on the first two columns reflect missing market data for 15DC and 15DP.}
    \label{tab:OOS2}
\end{table}

\begin{table}[H]
    \centering
        \begin{tabular}{cccc}
        &  &USD/EUR&\\
        &6m&9m&1y\\
        \hline\\
10DC&-0.0060 &  -0.0023 &   0.0015\\
25DC&-0.0034 &  -0.0019 &   0.0002\\
0   &-0.0018 &  -0.0020 &  -0.0017\\
25DP&-0.0006 &  -0.0010 &  -0.0005\\
10DP& 0.0002 &  -0.0002 &   0.0019\\
\hline\\
        &  &USD/JPY&\\
        &6m&9m&1y\\
        \hline\\
10DC&-0.0058 &  -0.0011 &   0.0000\\
15DC&-0.0042 &  -0.0005 &   0.0002\\
25DC&-0.0009 &   0.0005 &  -0.0000\\
0   &0.0016  &  0.0020  &  0.0000\\
25DP&0.0004  &  0.0009  &  0.0007\\
15DP&-0.0019 &  -0.0012 &   0.0002\\
10DP&-0.0037 &  -0.0020 &  -0.0002\\
\hline\\
        &  &EUR/JPY&\\
        &6m&9m&1y\\
        \hline\\
10DC&-0.0008 &  -0.0003 &  -0.0009\\
15DC& 0.0011 &   0.0007 &  -0.0006\\
25DC&0.0031  &  0.0015  & -0.0005\\
0   &0.0041  &  0.0025  &  0.0006\\
25DP&0.0014  &  0.0018  &  0.0026\\
15DP&-0.0022 &  -0.0004 &   0.0019\\
10DP&-0.0053 &  -0.0032 &   0.0003\\
\hline
        \end{tabular}
    \caption{Out-of-sample performance. This table reports the raw difference between the market implied volatility and the model generated implied volatility for 1 year, 9 and 6 months,  when we calibrate the model to the previous 3 expiries. Moneyness levels follow the standard Delta
quoting convention in the FX option market, see Footnote \ref{mktconv}. DC and DP stand for "delta call" and "delta put" respectively.}
    \label{tab:OOS3}
\end{table}

\begin{table}[H]
    \centering
        \begin{tabular}{ccccc}
        &  &USD/EUR&&\\
        &3m&6m&9m&1y\\
        \hline\\
10DC&-0.0062  & -0.0043 &  -0.0006  &  0.0031\\
25DC&-0.0027  & -0.0022 &  -0.0007  &  0.0010\\
0   &-0.0013  & -0.0009 &  -0.0014  & -0.0015\\
25DP&-0.0006  &  0.0001 &  -0.0005  & -0.0005\\
10DP& 0.0008  &  0.0009 &   0.0003  &  0.0020\\
\hline\\
        &  &USD/JPY&&\\
        &3m&6m&9m&1y\\
        \hline\\
10DC& -0.0090 &  -0.0052 &  -0.0004  &  0.0007\\
15DC& -0.0062 &  -0.0038 &  -0.0002  &  0.0003\\
25DC& -0.0037 &  -0.0009 &   0.0003  & -0.0006\\
0   & 0.0010  &  0.0014  &  0.0015   & -0.0011\\
25DP&-0.0004  &  0.0006  &  0.0009   & 0.0002\\
15DP&-0.0022  & -0.0014  & -0.0008   & 0.0002\\
10DP&-0.0039  & -0.0029  & -0.0013   & 0.0003\\
\hline\\
        &  &EUR/JPY&&\\
        &3m&6m&9m&1y\\
        \hline\\
10DC&-0.0045  & -0.0003  &  0.0002  & -0.0006\\
15DC&-0.0021  &  0.0014  &  0.0009  & -0.0007\\
25DC& 0.0010  &  0.0030  &  0.0012  & -0.0013\\
0   & 0.0030  &  0.0038  &  0.0019  & -0.0008\\
25DP&  0.0009 &   0.0014 &   0.0014 &   0.0016\\
15DP&-0.0026  & -0.0020  & -0.0005  &  0.0013\\
10DP&-0.0057  & -0.0050  & -0.0030  & -0.0000\\
\hline
        \end{tabular}
    \caption{Out-of-sample performance. This table reports the raw difference between the market implied volatility and the model generated implied volatility for 1 year, 9, 6 and 3 months,  when we calibrate the model to the previous 2 expiries. Moneyness levels follow the standard Delta
quoting convention in the FX option market, see Footnote \ref{mktconv}. DC and DP stand for "delta call" and "delta put" respectively.}
    \label{tab:OOS4}
\end{table}

\begin{table}[H]
    \centering
        \begin{tabular}{cccc}
        &USD/EUR&USD/JPY&EUR/JPY\\
        \hline\\
10DC&0.0058 &   0.0030&    0.0046\\
15DC&       &   0.0022&    0.0035\\
25DC&0.0058 &   0.0009&    0.0041\\
0   &0.0035 &   0.0007&    0.0019\\
25DP&0.0029 &   0.0024&    0.0016\\
15DP&       &   0.0027&   -0.0012\\
10DP&0.0042 &   0.0029&   -0.0044\\
\hline
        \end{tabular}
    \caption{Out-of-sample performance. This table reports the raw difference between the market implied volatility and the model generated implied volatility for 1 year, when we calibrate the model to the previous 5 expiries and $\kappa_k=1,k=1,2$. Moneyness levels follow the standard Delta
quoting convention in the FX option market, see Footnote \ref{mktconv}. DC and DP stand for "delta call" and "delta put" respectively. Blanks on the first column reflect missing market data for 15DC and 15DP.}
    \label{tab:OOS1B}
\end{table}

\begin{table}[H]
    \centering
        \begin{tabular}{ccccccc}
        &USD/EUR&USD/EUR&USD/JPY&USD/JPY&EUR/JPY&EUR/JPY\\
        &9m&1y&9m&1y&9m&1y\\
        \hline\\
10DC&    0.0083&    0.0092&    0.0034&    0.0052&    0.0073&    0.0090\\
15DC&          &          &    0.0021&    0.0042&    0.0061&    0.0078\\
25DC&    0.0060&    0.0092&    0.0005&    0.0029&    0.0053&    0.0079\\
0   &    0.0029&    0.0066&   -0.0000&    0.0025&    0.0024&    0.0050\\
25DP&    0.0033&    0.0056&    0.0026&    0.0044&    0.0025&    0.0041\\
15DP&          &          &    0.0034&    0.0049&    0.0006&    0.0010\\
10DP&    0.0062&    0.0067&    0.0041&    0.0054&   -0.0027&   -0.0024\\
\hline
        \end{tabular}
    \caption{Out-of-sample performance. This table reports the raw difference between the market implied volatility and the model generated implied volatility for 1 year and 9 months, when we calibrate the model to the previous 4 expiries and $\kappa_k=1,k=1,2$. Moneyness levels follow the standard Delta
quoting convention in the FX option market, see Footnote \ref{mktconv}. DC and DP stand for "delta call" and "delta put" respectively. Blanks on the first two columns reflect missing market data for 15DC and 15DP.}
    \label{tab:OOS2B}
\end{table}

\begin{table}[H]
    \centering
        \begin{tabular}{cccc}
        &  &USD/EUR&\\
        &6m&9m&1y\\
        \hline\\
10DC&0.0073  &  0.0116  &  0.0131\\
25DC&0.0042  &  0.0091  &  0.0128\\
0   &0.0010  &  0.0056  &  0.0100\\
25DP&0.0020  &  0.0058  &  0.0087\\
10DP&0.0056  &  0.0088  &  0.0099\\
\hline \\
        &  &USD/JPY&\\
        &6m&9m&1y\\
        \hline\\
10DC&0.0025  &  0.0048  &  0.0067\\
15DC&0.0009  &  0.0030  &  0.0053\\
25DC&-0.0013 &   0.0008 &   0.0034\\
0   &-0.0028 &  -0.0004 &   0.0023\\
25DP&0.0002  &  0.0027  &  0.0048\\
15DP&0.0013  &  0.0041  &  0.0057\\
10DP&0.0022  &  0.0053  &  0.0067\\
\hline \\
        &  &EUR/JPY&\\
        &6m&9m&1y\\
        \hline\\
10DC&0.0051 &   0.0114  &  0.0138\\
15DC&0.0036 &   0.0099  &  0.0123\\
25DC&0.0022 &   0.0085  &  0.0120\\
0   &-0.0001&    0.0048 &   0.0083\\
25DP&0.0017 &   0.0047  &  0.0071\\
15DP&0.0010 &   0.0028  &  0.0040\\
10DP&-0.0016&   -0.0003 &   0.0006\\
\hline
        \end{tabular}
    \caption{Out-of-sample performance. This table reports the raw difference between the market implied volatility and the model generated implied volatility for 1 year, 9 and 6 months,  when we calibrate the model to the previous 3 expiries and $\kappa_k=1,k=1,2$. Moneyness levels follow the standard Delta
quoting convention in the FX option market, see Footnote \ref{mktconv}. DC and DP stand for "delta call" and "delta put" respectively.}
    \label{tab:OOS3B}
\end{table}

\begin{table}[H]
    \centering
        \begin{tabular}{ccccc}
        &  &USD/EUR&&\\
        &3m&6m&9m&1y\\
        \hline\\
10DC&0.0057 &   0.0085  &  0.0125&    0.0136\\
25DC&0.0023 &   0.0042  &  0.0086&    0.0121\\
0   &-0.0018&   -0.0000 &   0.0041&    0.0082\\
25DP&-0.0008&    0.0010 &   0.0043&    0.0068\\
10DP&0.0025 &   0.0050  &  0.0077 &   0.0086\\
\hline \\
        &  &USD/JPY&&\\
        &3m&6m&9m&1y\\
        \hline\\
10DC&0.0013  &  0.0026 &   0.0046  &  0.0063\\
15DC&0.0004  &  0.0003 &   0.0020  &  0.0040\\
25DC&-0.0015 &  -0.0028&   -0.0012 &   0.0010\\
0   &-0.0033 &  -0.0051&   -0.0033 &  -0.0010   \\
25DP&-0.0010 &  -0.0015&    0.0005 &   0.0022\\
15DP&-0.0004 &   0.0002&    0.0025 &   0.0038\\
10DP&-0.0002 &   0.0017&    0.0043 &   0.0053\\
\hline \\
        &  &EUR/JPY&&\\
        &3m&6m&9m&1y\\
        \hline\\
10DC&0.0002  &  0.0048  &  0.0106  &  0.0127\\
15DC&-0.0009 &   0.0025 &   0.0082 &   0.0101\\
25DC& -0.0026&   -0.0001&    0.0056&    0.0086\\
0   & -0.0038&   -0.0035&    0.0008&    0.0037\\
25DP&-0.0008 &  -0.0007 &   0.0016 &   0.0034\\
15DP&-0.0010 &  -0.0007 &   0.0005 &   0.0012\\
10DP&-0.0021 &  -0.0026 &  -0.0020 &  -0.0015\\
\hline
        \end{tabular}
    \caption{Out-of-sample performance. This table reports the raw difference between the market implied volatility and the model generated implied volatility for 1 year, 9, 6 and 3 months,  when we calibrate the model to the previous 2 expiries and $\kappa_k=1,k=1,2$. Moneyness levels follow the standard Delta
quoting convention in the FX option market, see Footnote \ref{mktconv}. DC and DP stand for "delta call" and "delta put" respectively.}
    \label{tab:OOS4B}
\end{table}

\begin{table}[H]
    \centering
        \begin{tabular}{cccccc}
&6&5&4&3&2\\
\hline\\
$k=1$   &-0.1715  & -0.1841 &  -0.2076 &  -0.2276 &  -0.2445\\
$k=2$ &-0.6745  & -0.6640 &  -0.5086 &  -0.5153 &  -0.5768\\
\hline
        \end{tabular}
            \caption{For all $k=1,2$ and for each sample we report the quantity $2\kappa_k\theta_k-\xi_k^{2}$.
            In all cases the quantity is negative and its absolute value is a measure of the violation of the Feller condition.}
    \label{tab:Feller}
\end{table}

\begin{table}[H]
    \centering
        \begin{tabular}{cccc}
Order&$S^{\scriptscriptstyle\mathrm{USD,EUR}}$&$S^{\scriptscriptstyle\mathrm{JPY,USD}}$&$S^{\scriptscriptstyle\mathrm{JPY,EUR}}$\\
\hline
\hline
1& $+\infty$ & $+\infty$ & $+\infty$\\
2& $+\infty$ &  12.1962  &  5.1612\\
3& $+\infty$ &   2.9537  &  2.0580\\
4&    3.3968 &   1.7990  &  1.3763\\
5&    2.0070 &   2.0819  &  1.0614\\
\hline
        \end{tabular}
    \caption{Times of moment explosions for moments up to order 5 for the three currency pairs. First moments are always finite.}
    \label{tab:expl}
\end{table}

\begin{figure}[htbp]
    \centering
        \includegraphics[scale=0.13]{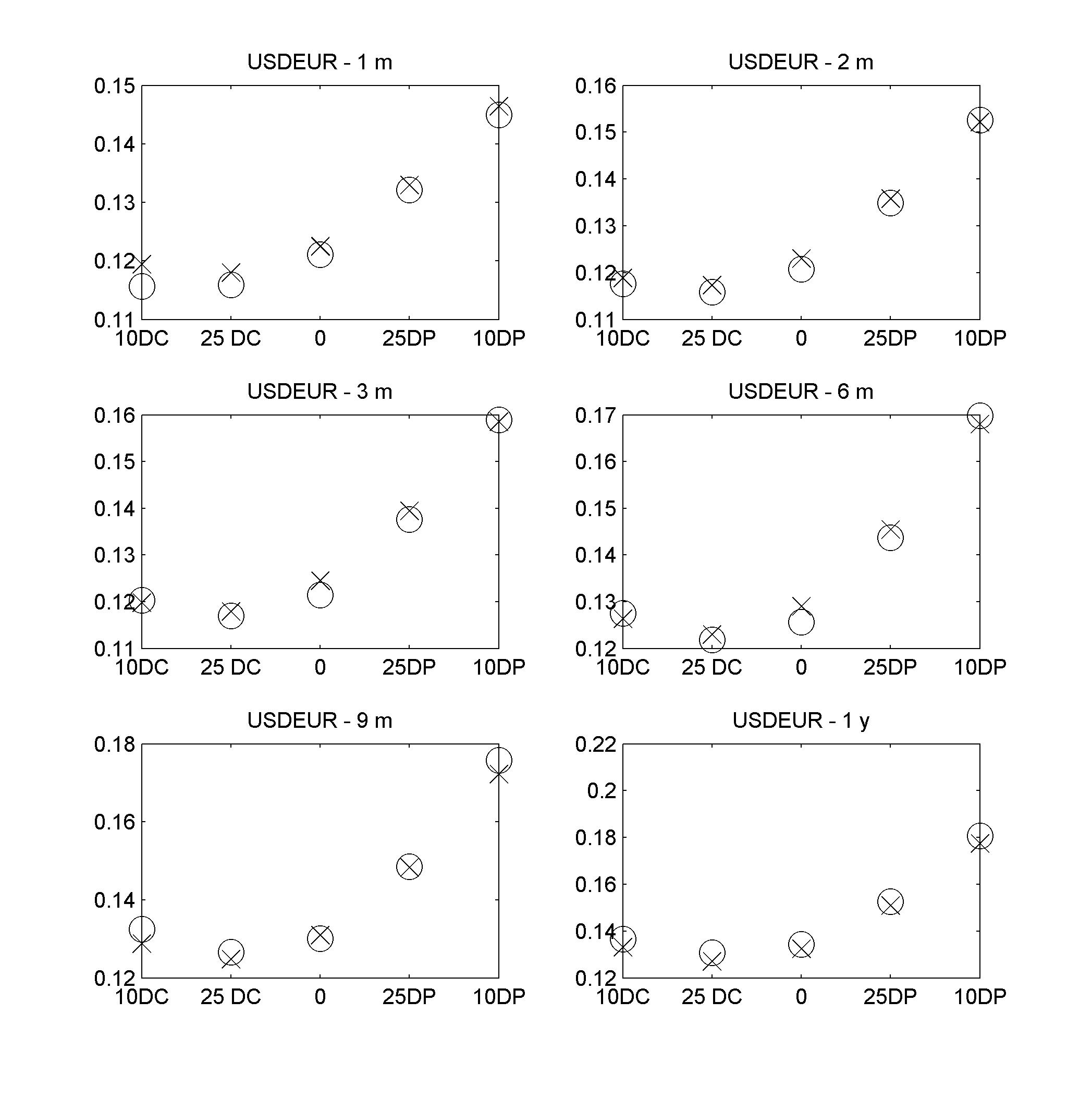}
    \caption{Calibration of the USD/EUR implied volatility surface. Market data as of 23/7/2010. Market volatilities are denoted by crosses,
model volatilities are denoted by circles. Moneyness levels follow the standard Delta
quoting convention in the FX option market, see Footnote \ref{mktconv}. DC and DP stand for "delta call" and "delta put" respectively.}
    \label{fig:USDEUR}
\end{figure}

\begin{figure}[htbp]
    \centering
        \includegraphics[scale=0.13]{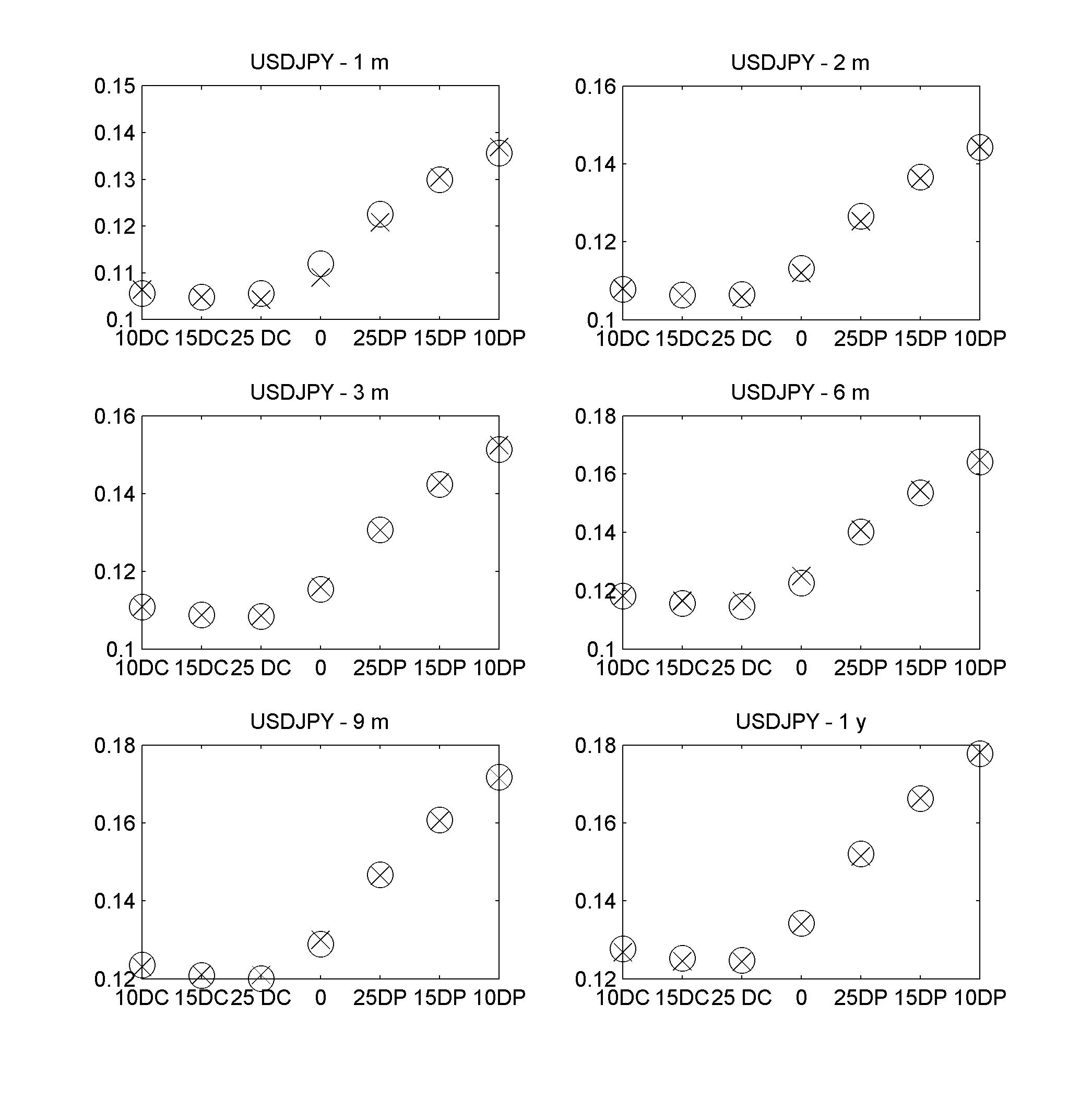}
    \caption{Calibration of the USD/JPY implied volatility surface. Market data as of 23/7/2010.  Market volatilities are denoted by crosses,
model volatilities are denoted by circles. Moneyness levels follow the standard Delta
quoting convention in the FX option market, see Footnote \ref{mktconv}. DC and DP stand for "delta call" and "delta put" respectively.}
    \label{fig:USDJPY}
\end{figure}

\begin{figure}[htbp]
    \centering
        \includegraphics[scale=0.13]{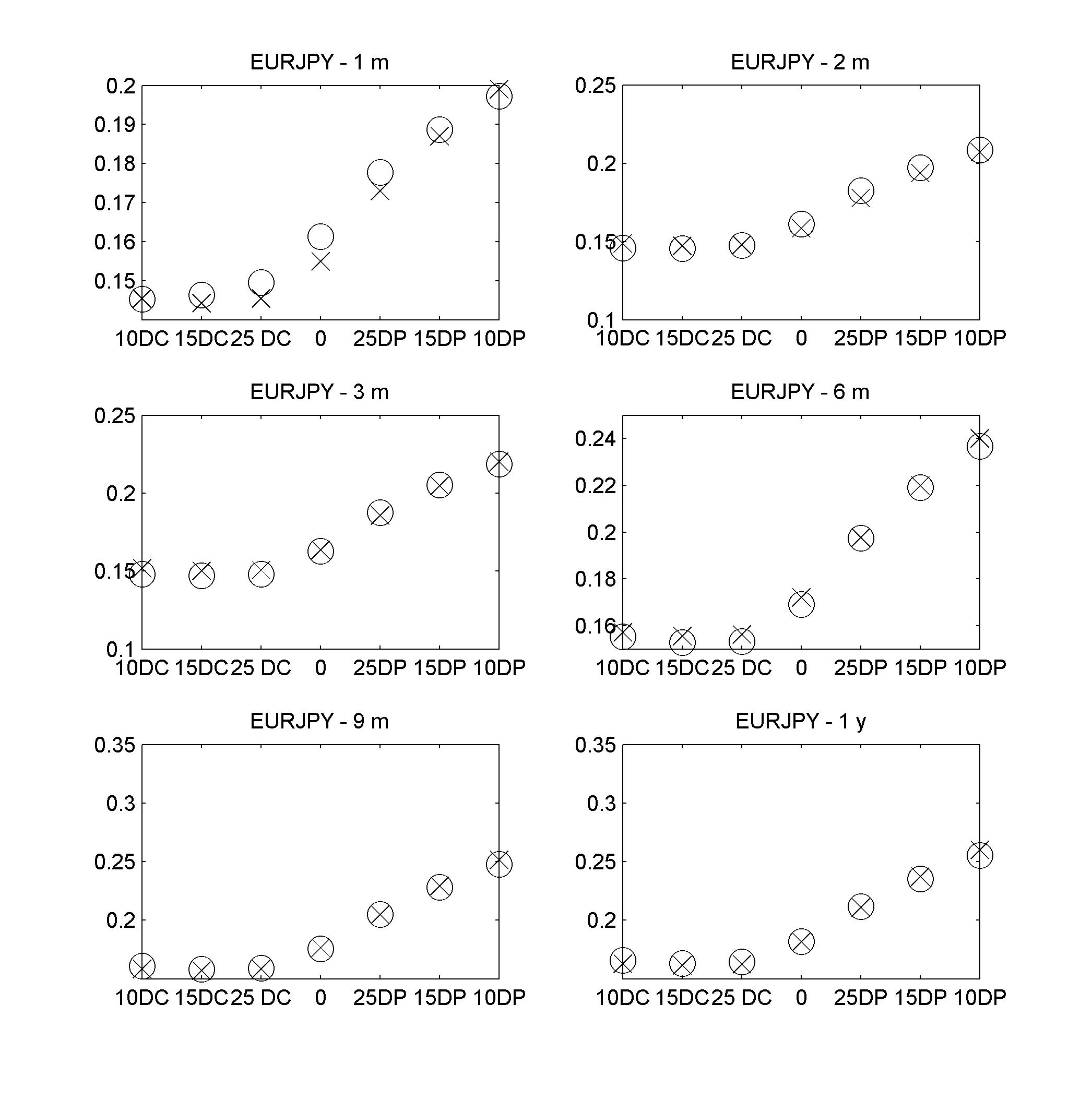}
    \caption{Calibration of the EUR/JPY implied volatility surface. Market data as of 23/7/2010.  Market volatilities are denoted by crosses,
model volatilities are denoted by circles. Moneyness levels follow the standard Delta
quoting convention in the FX option market, see Footnote \ref{mktconv}. DC and DP stand for "delta call" and "delta put" respectively.}
    \label{fig:EURJPY}
\end{figure}

\begin{figure}[htbp]
\centering
\subfloat{\label{fig:11}\includegraphics[scale=0.09]{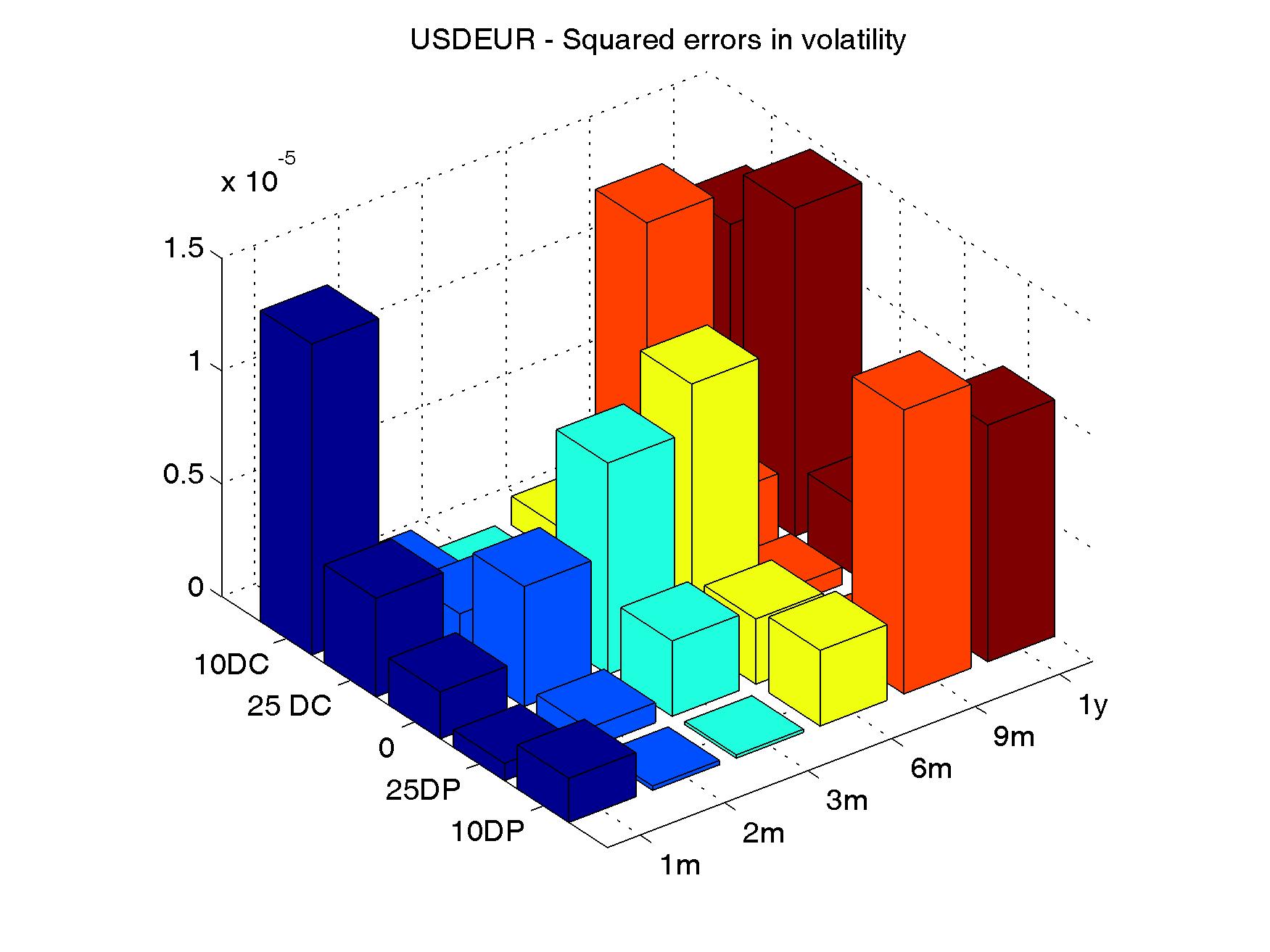}}\\
\subfloat{\label{fig:12}\includegraphics[scale=0.09]{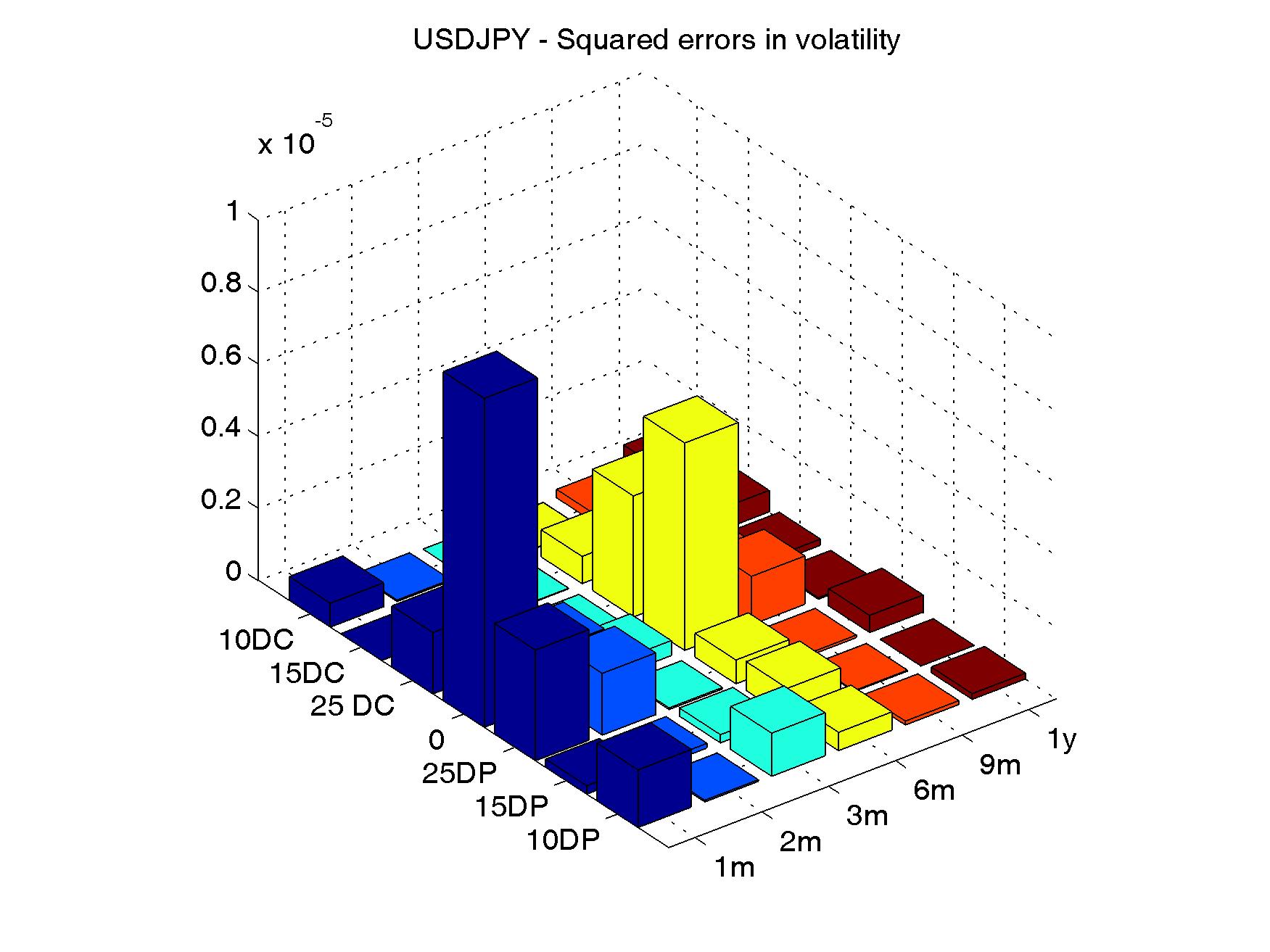}}
\subfloat{\label{fig:13}\includegraphics[scale=0.09]{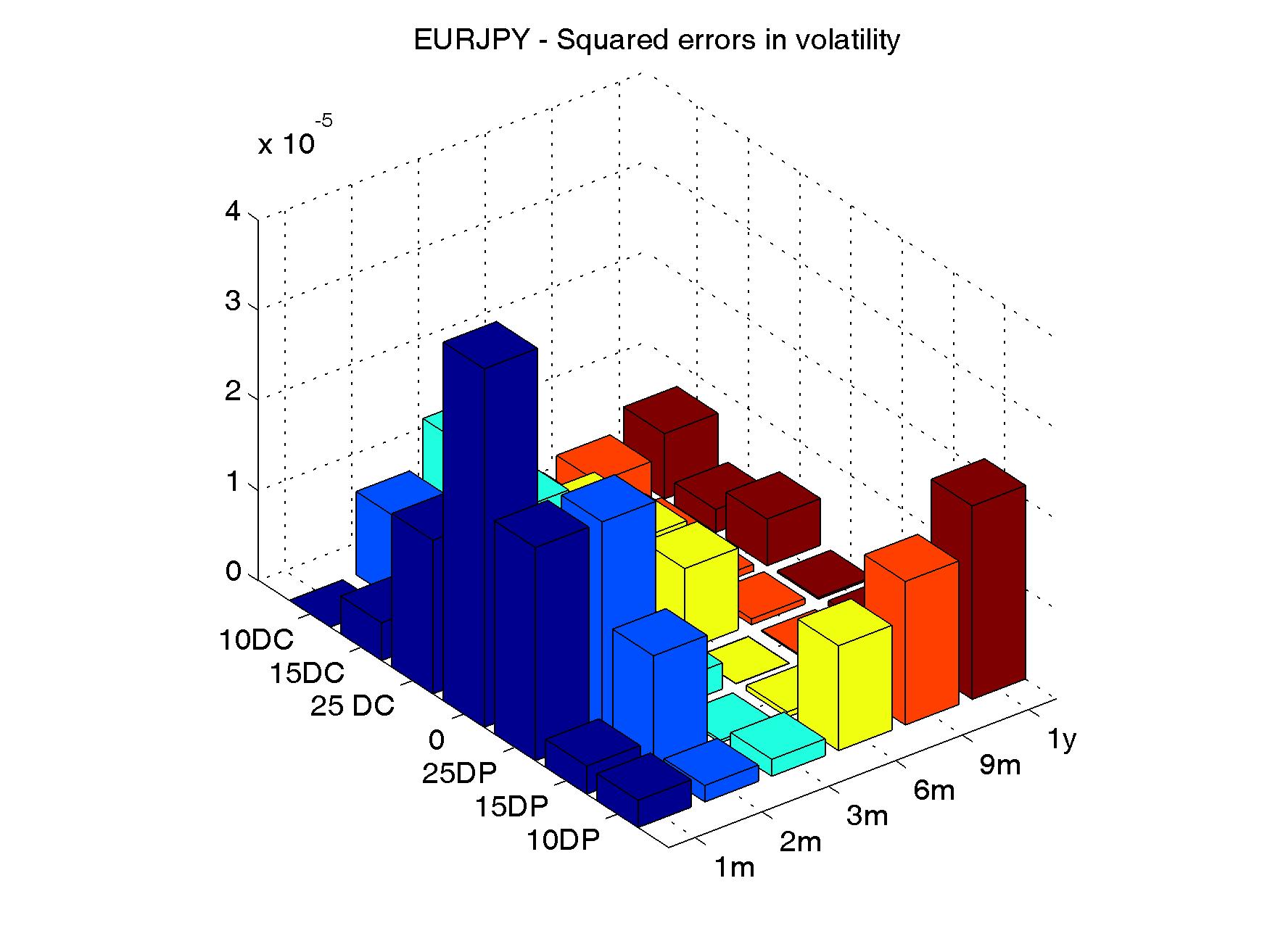}}
\caption{In-sample implied volatilities squared error for the joint calibration on 1m, 2m, 3m, 6m, 9m and 1y. The associated model parameters may be found in Table \ref{tab:params}, column "6".}
\label{fig:SE_1}
\end{figure}

\begin{figure}[htbp]
\centering
\subfloat{\label{fig:21}\includegraphics[scale=0.09]{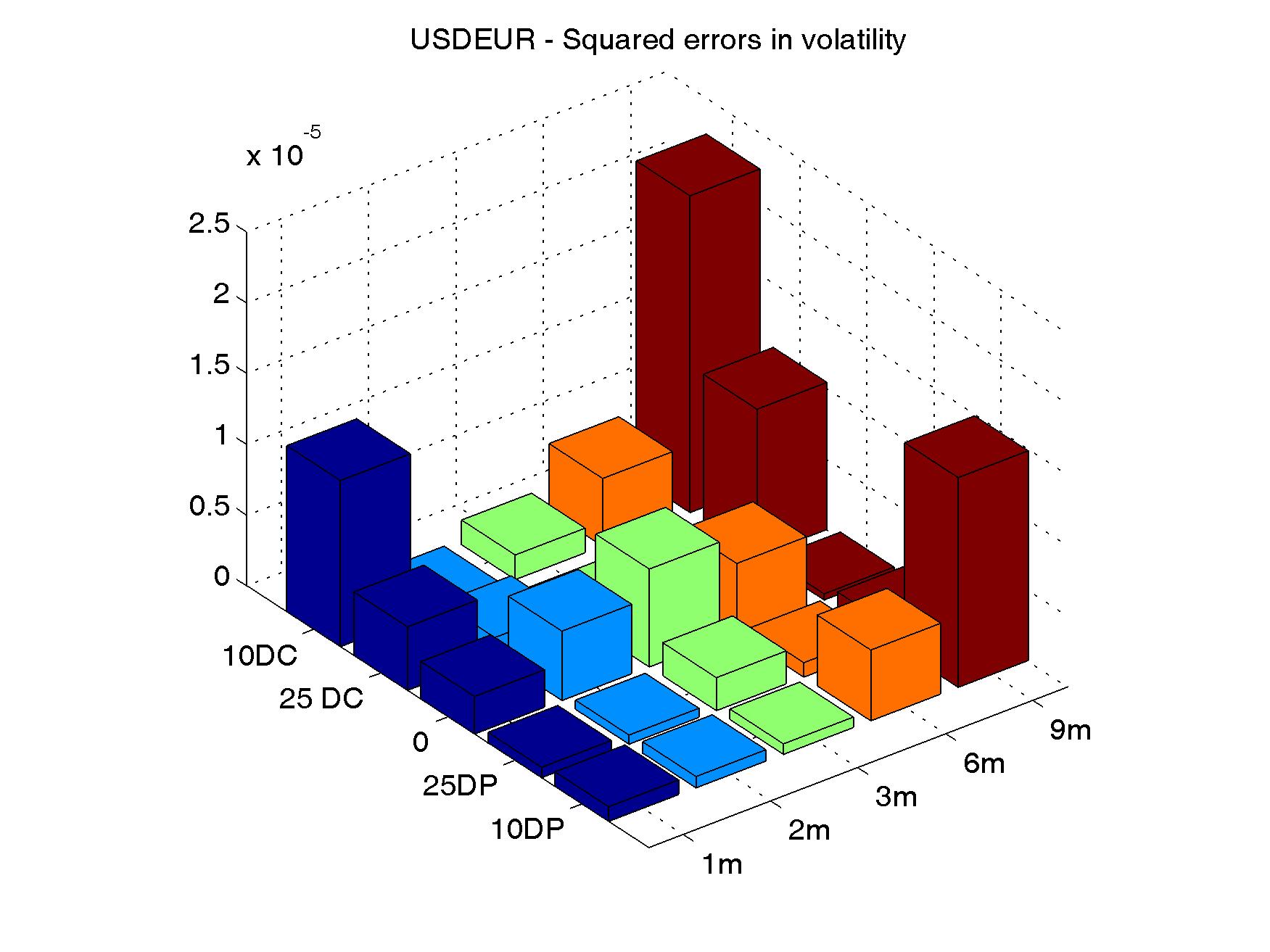}}\\
\subfloat{\label{fig:22}\includegraphics[scale=0.09]{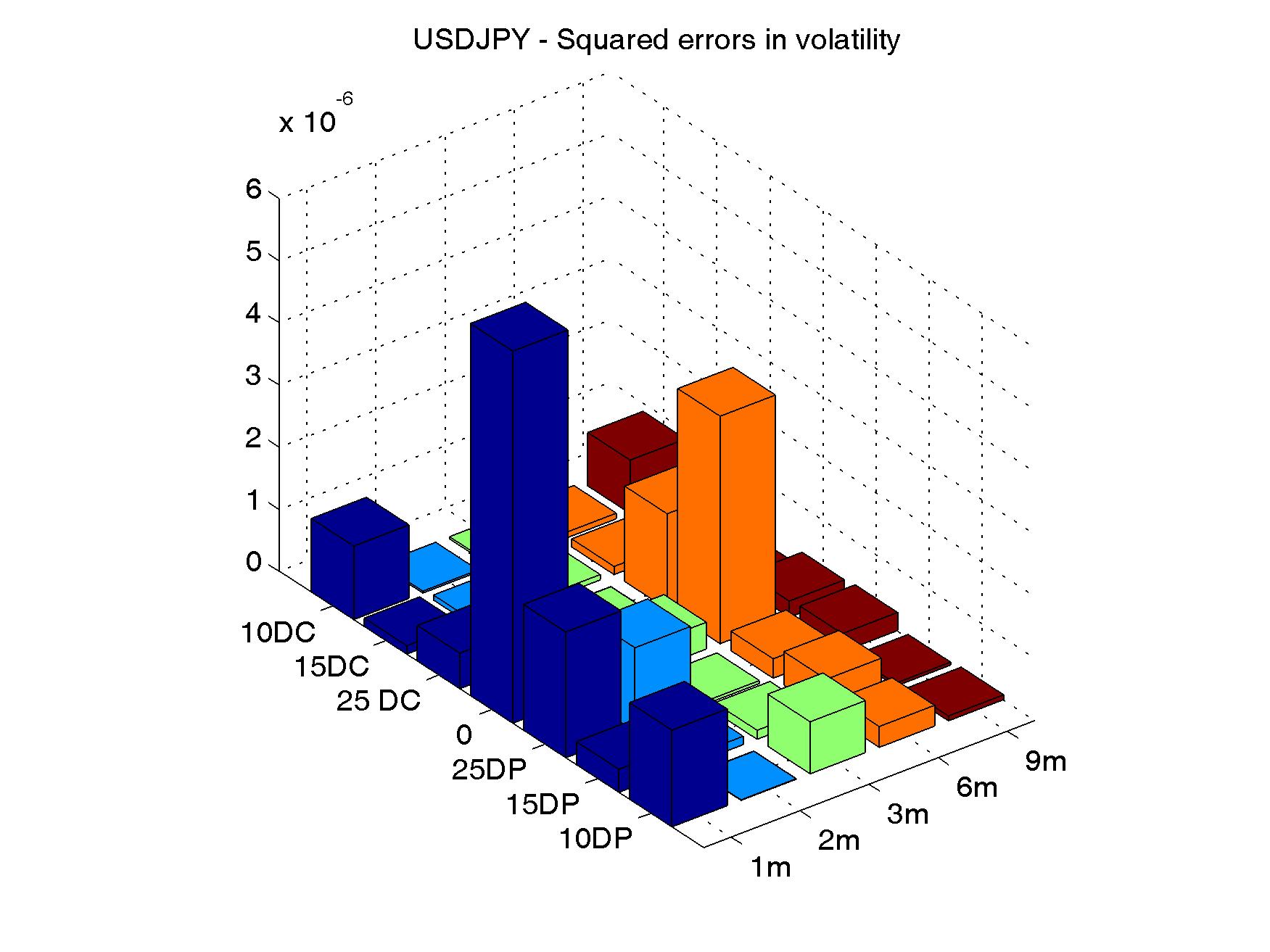}}
\subfloat{\label{fig:23}\includegraphics[scale=0.09]{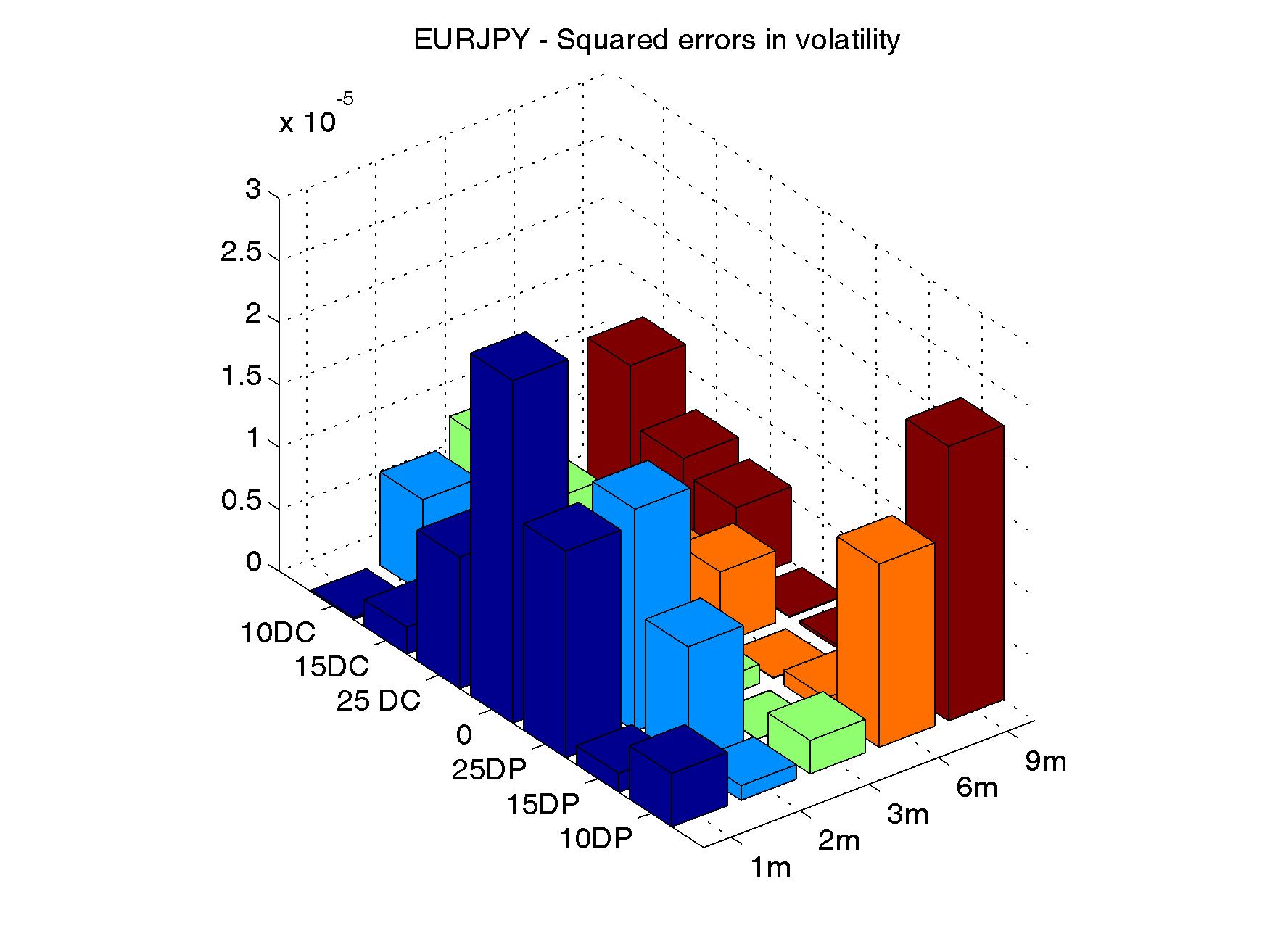}}
\caption{In-sample implied volatilities squared error for the joint calibration on 1m, 2m, 3m, 6m, and 9m. The associated model parameters may be found in Table \ref{tab:params}, column "5".}
\label{fig:SE_2}
\end{figure}

\begin{figure}[htbp]
\centering
\subfloat{\label{fig:31}\includegraphics[scale=0.09]{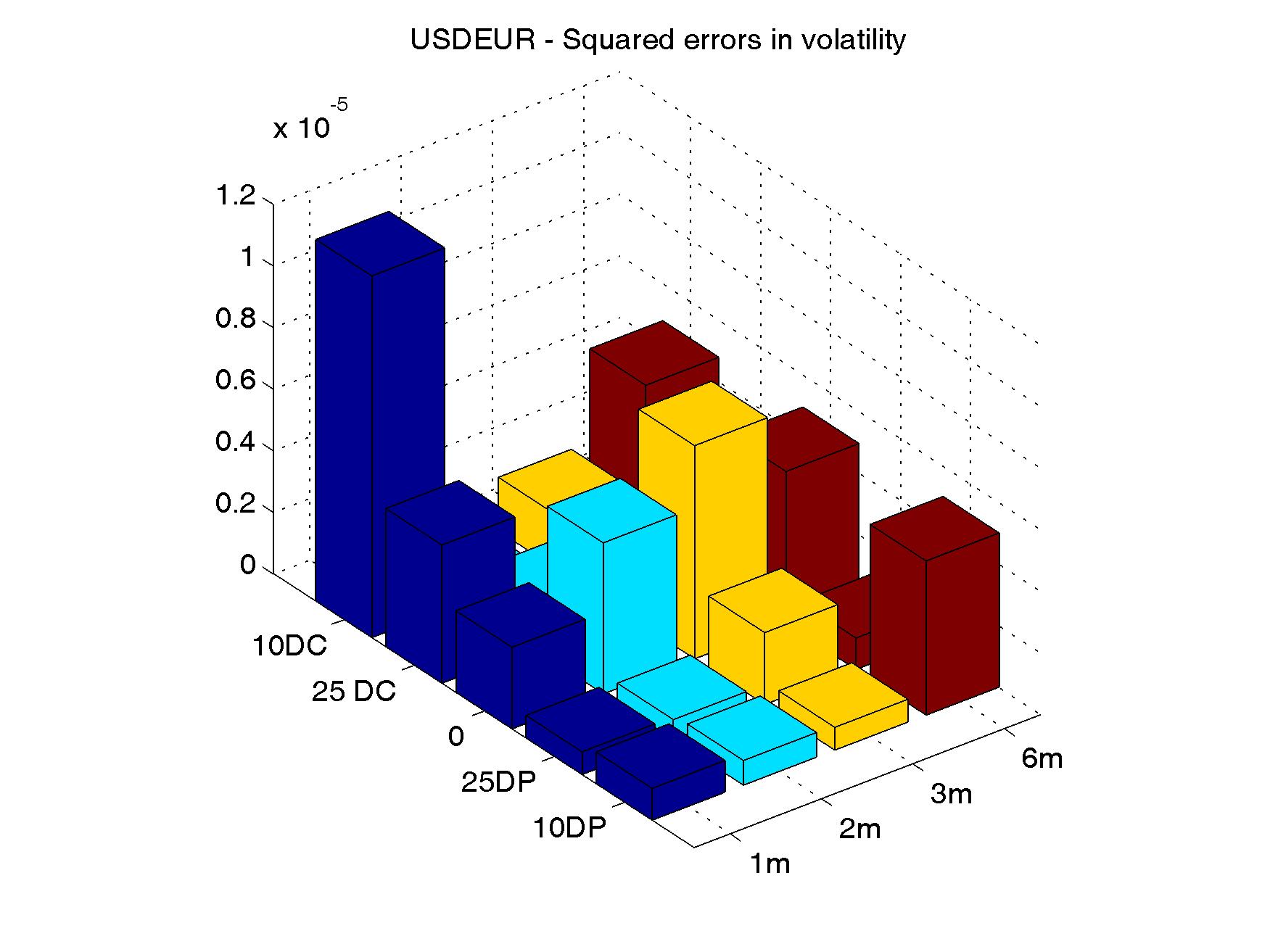}}\\
\subfloat{\label{fig:32}\includegraphics[scale=0.09]{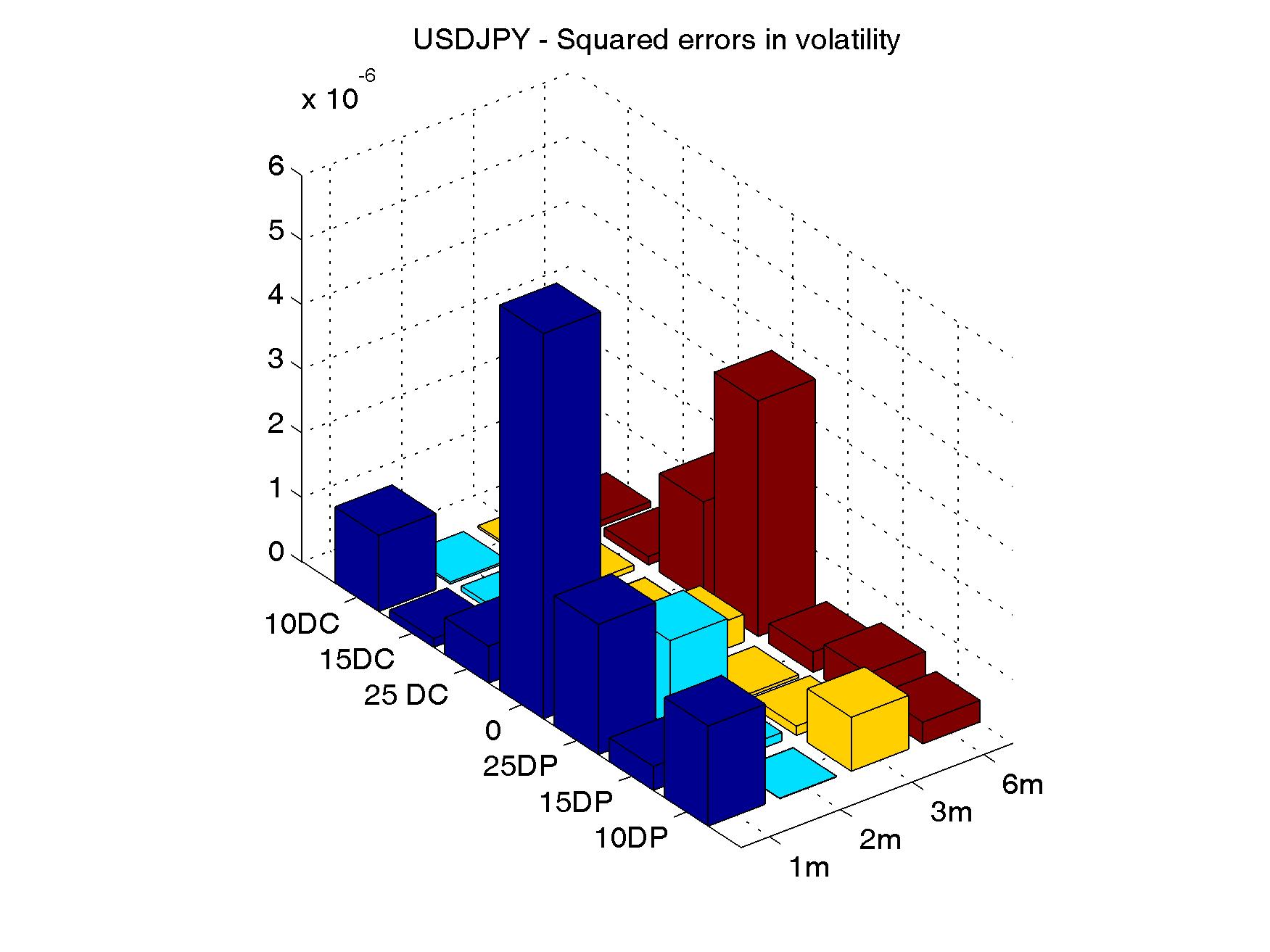}}
\subfloat{\label{fig:33}\includegraphics[scale=0.09]{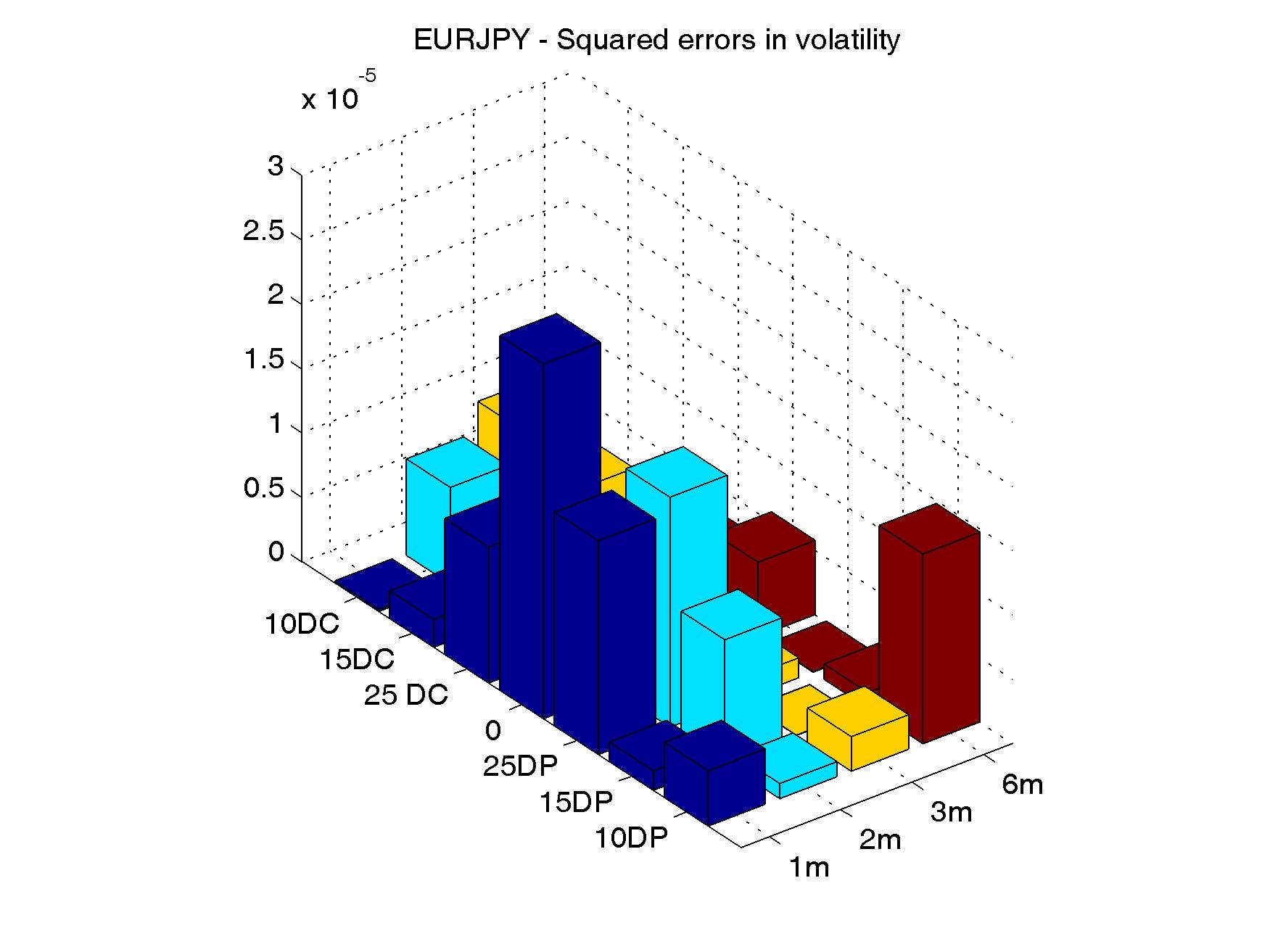}}
\caption{In-sample implied volatilities squared error for the joint calibration on 1m, 2m, 3m, and 6m. The associated model parameters may be found in Table \ref{tab:params}, column "4".}
\label{fig:SE_3}
\end{figure}

\begin{figure}[htbp]
\centering
\subfloat{\label{fig:41}\includegraphics[scale=0.09]{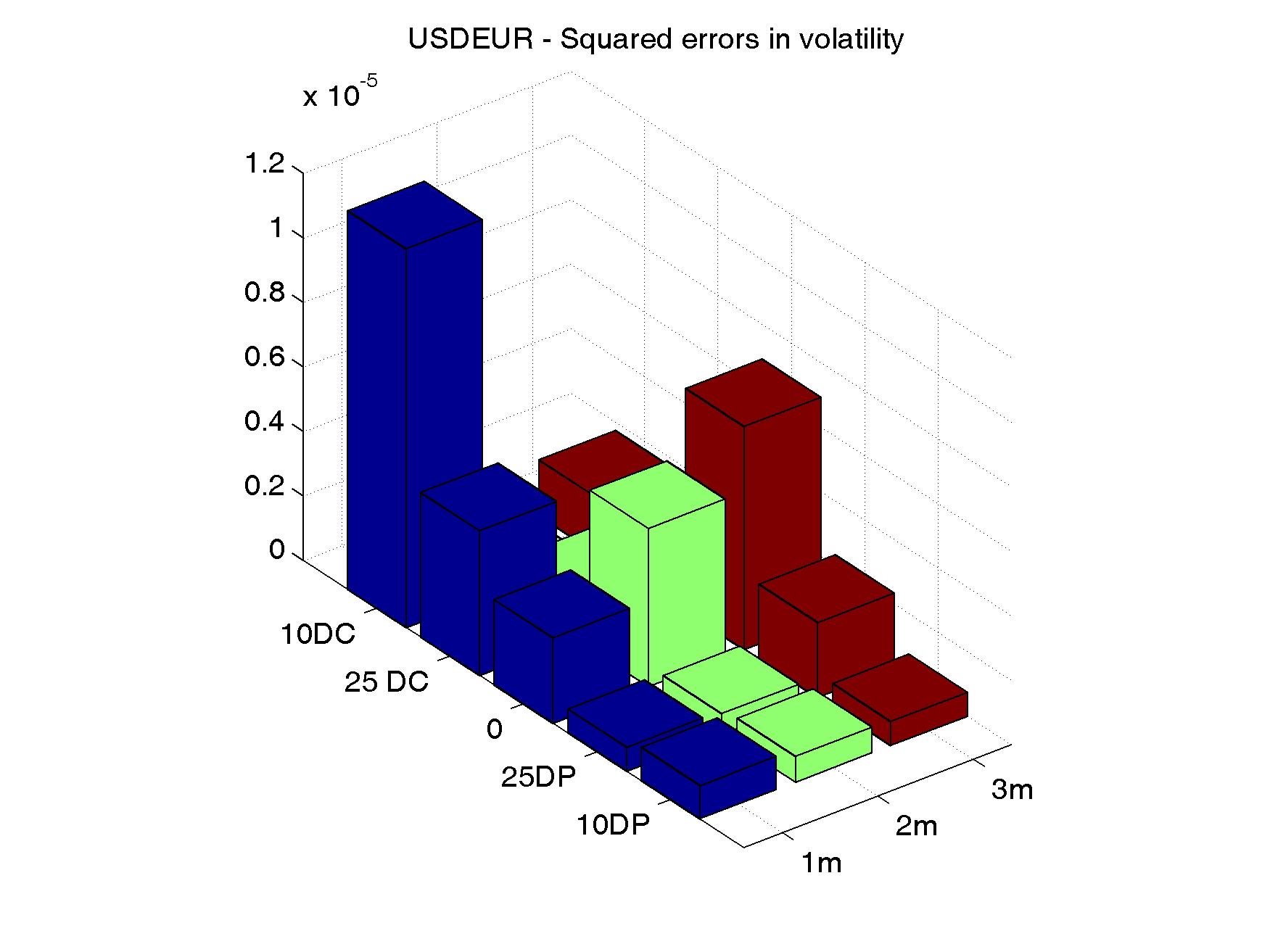}}\\
\subfloat{\label{fig:42}\includegraphics[scale=0.09]{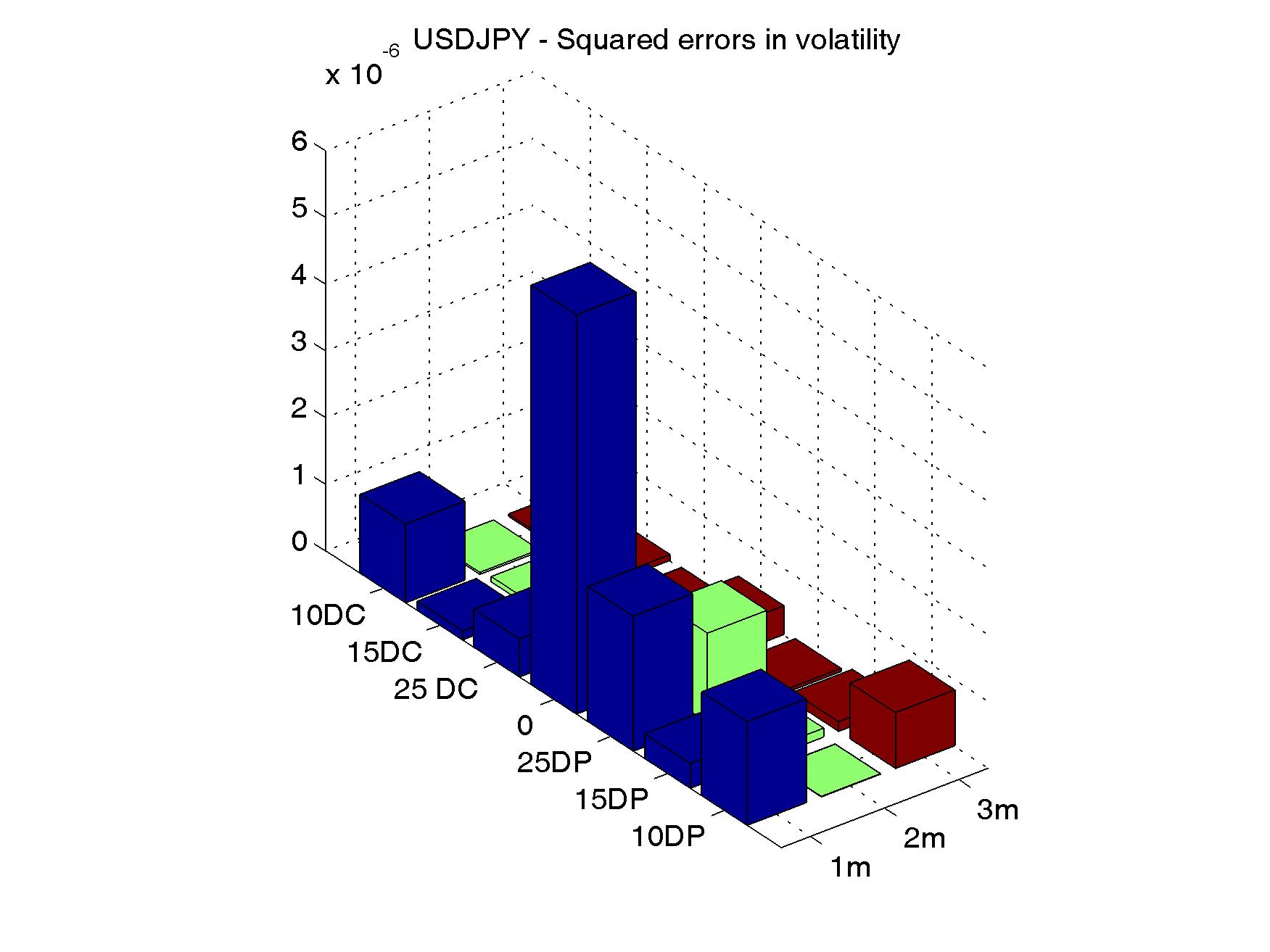}}
\subfloat{\label{fig:43}\includegraphics[scale=0.09]{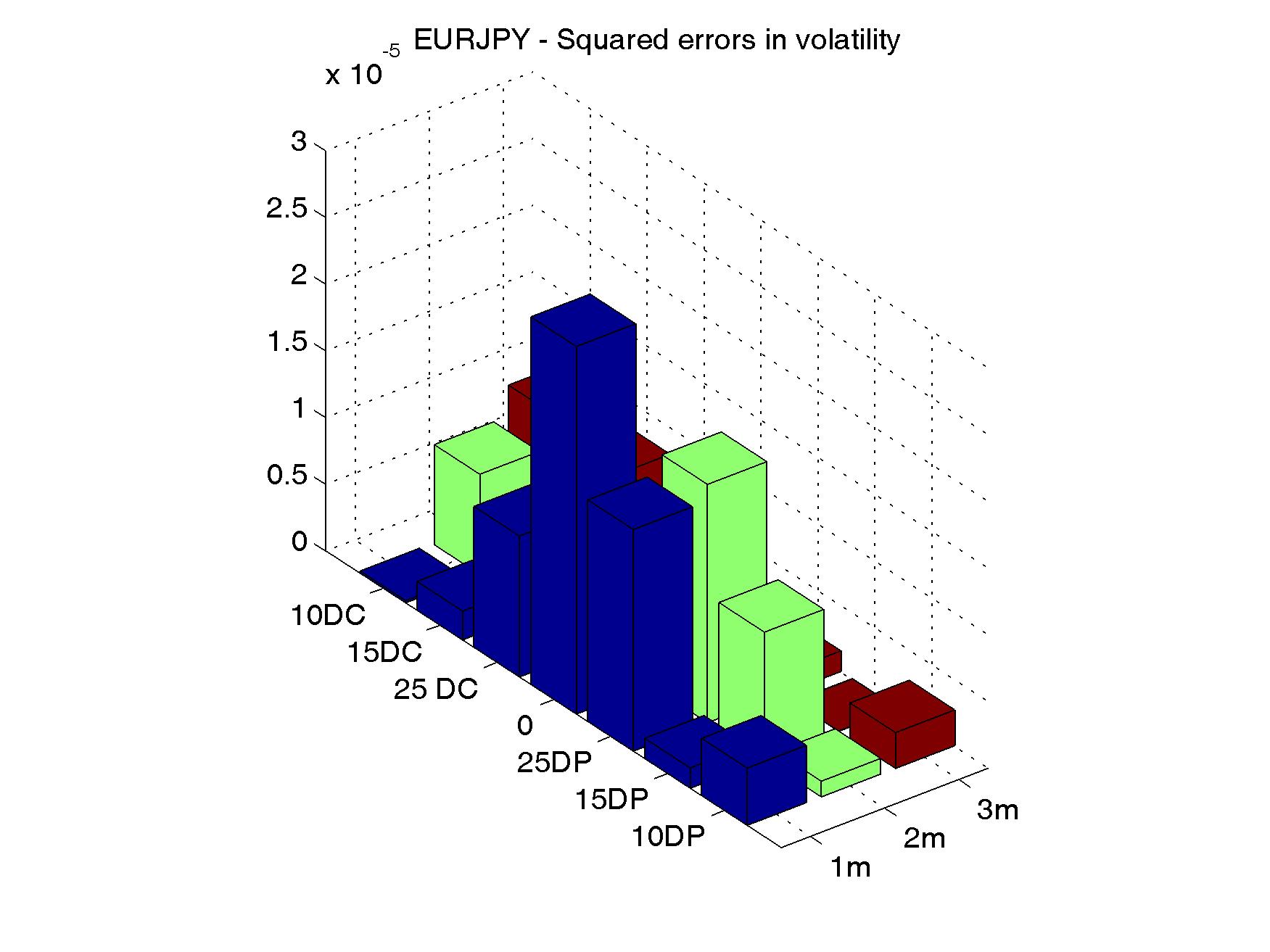}}
\caption{In-sample implied volatilities squared error for the joint calibration on 1m, 2m, and 3m. The associated model parameters may be found in Table \ref{tab:params}, column "3".}
\label{fig:SE_4}
\end{figure}

\begin{figure}[htbp]
\centering
\subfloat{\label{fig:51}\includegraphics[scale=0.09]{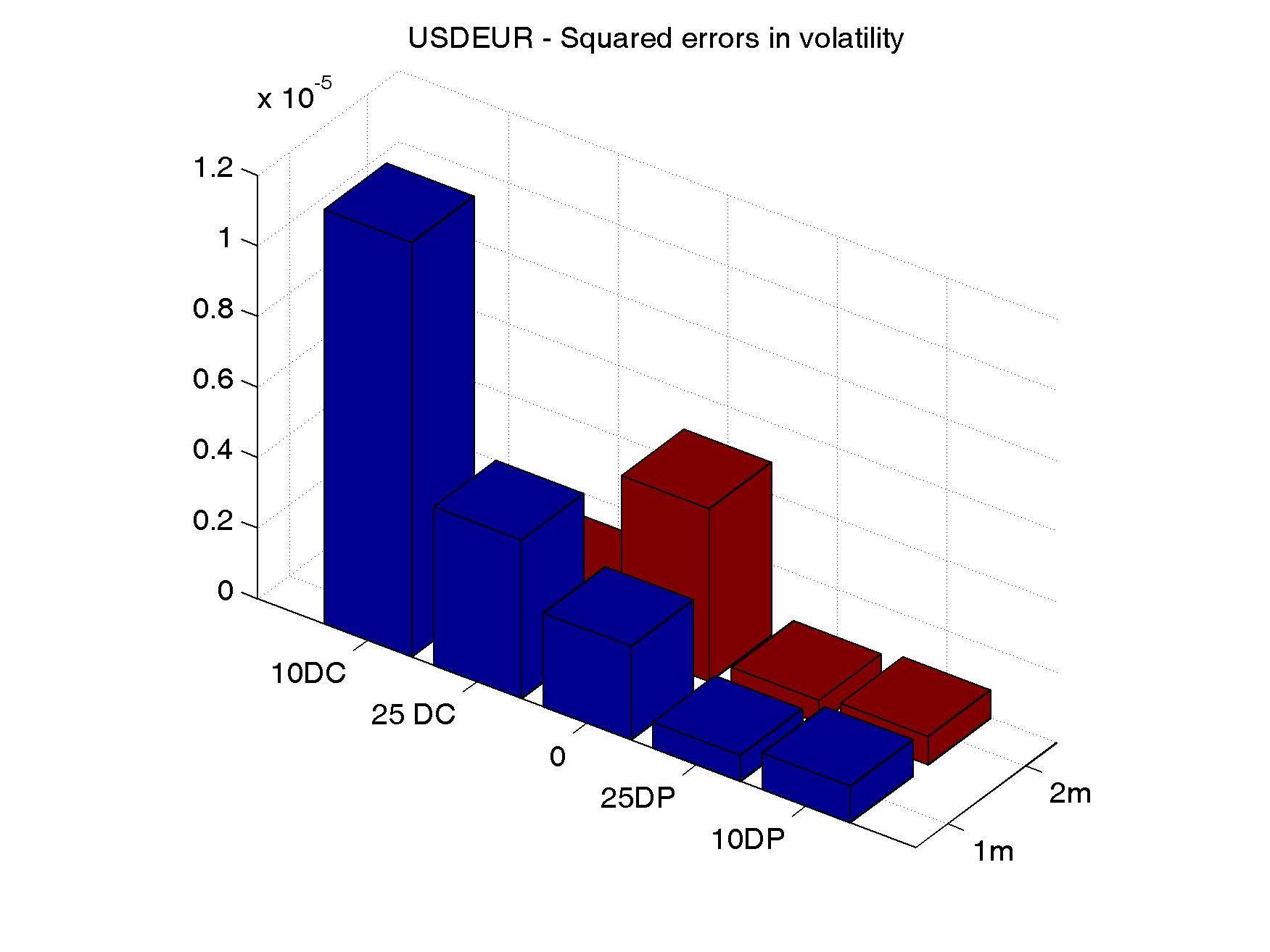}}\\
\subfloat{\label{fig:52}\includegraphics[scale=0.09]{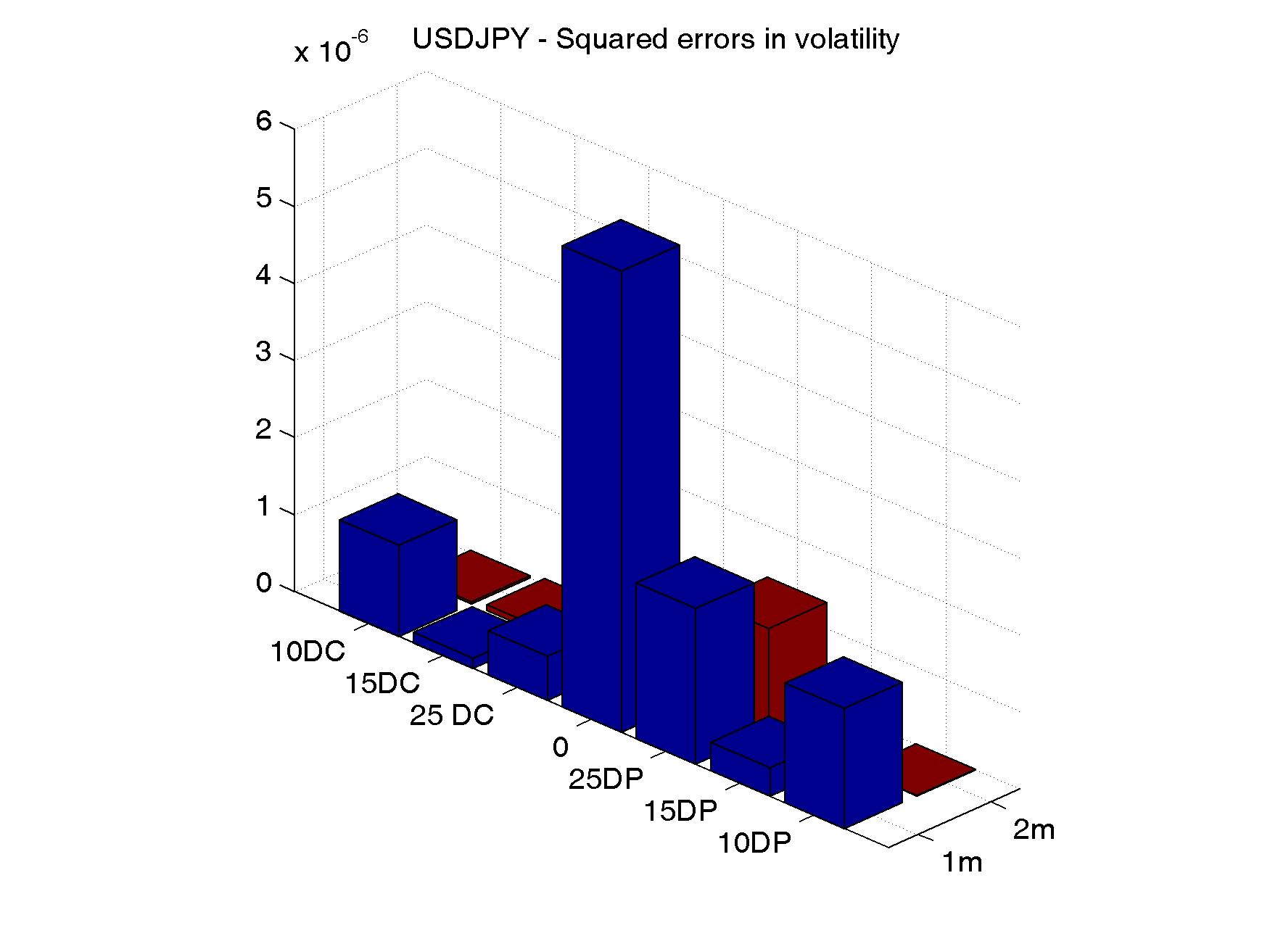}}
\subfloat{\label{fig:53}\includegraphics[scale=0.09]{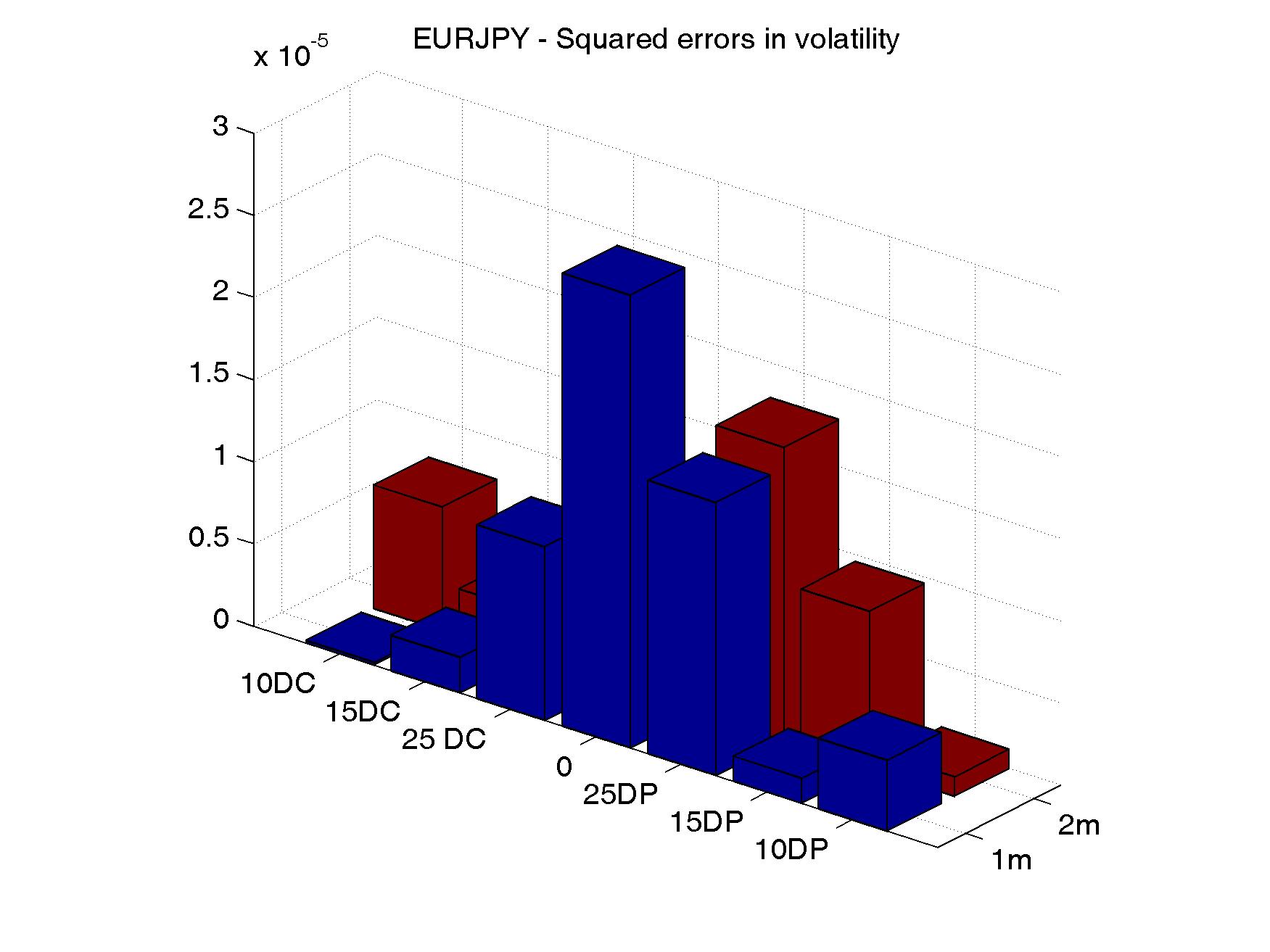}}
\caption{In-sample implied volatilities squared error for the joint calibration on 1m, and 2m. The associated model parameters may be found in Table \ref{tab:params}, column "2".}
\label{fig:SE_5}
\end{figure}

\newpage
\subsection{Calibration of AUD/USD/JPY}
\begin{table}[H]
    \centering
        \begin{tabular}{ccccccc}
&6&5&4&3&2\\
\hline\\
$V_1$& 0.0115&    0.0117&    0.0109&    0.0124&    0.0133\\
$V_2$& 0.0176&    0.0159&    0.0100&    0.0196&    0.0171\\
$a^{\scriptscriptstyle\mathrm{USD}}_1$& 0.5142&    0.5204&    0.5097&    0.5266&    0.5184\\
$a^{\scriptscriptstyle\mathrm{USD}}_2$& 1.3692&    1.3757&    1.4358&    1.4075&    1.3909\\
$a^{\scriptscriptstyle\mathrm{AUD}}_1$& 0.8361&    0.8179&    0.7874&    0.7681&    0.7516\\
$a^{\scriptscriptstyle\mathrm{AUD}}_2$& 0.7999&    0.7789&    0.6823&    0.8455&    0.7657\\
$a^{\scriptscriptstyle\mathrm{JPY}}_1$& 1.2029&    1.2047&    1.2211&    1.2159&    1.2104\\
$a^{\scriptscriptstyle\mathrm{JPY}}_2$& 1.5020&    1.4957&    1.5534&    1.4756&    1.4596\\
$\kappa_1$& 1.9865&    1.2988&    0.2464&    2.0000&    2.0000\\
$\kappa_2$& 0.8134&    0.2298&    0.0647&    1.2270&    1.8346\\
$\theta_1$& 0.0314&    0.0398&    0.1351&    0.0249&    0.0182\\
$\theta_2$& 0.0823&    0.2332&    0.5798&    0.0608&    0.0327\\
$\xi_1$& 0.7196&    0.6358&    0.5264&    0.5793&    0.6086\\
$\xi_2$& 1.0000&    0.8956&    0.7317&    1.0000&    1.0000\\
$\rho_1$& 0.3337&    0.3384&    0.3408&    0.3225&    0.3283\\
$\rho_2$& -0.4451&   -0.4455&   -0.4416&   -0.3910&   -0.3834\\
Res. Norm. &  8.2178e-04&7.6368e-04& 0.0016& 0.0010& 5.4606e-05\\
\hline
        \end{tabular}
    \caption{This table reports the results of the calibration of the model.
    We concentrate on the two factor case. For each column, a different number of expiries,
    ranging from 6 to 2, is chosen. More specifically, 6 means that the following expiries are
    considered: 1, 2, 3, 6, 9 months and 1 year, whereas 5 means that the longest maturity, i.e. 1
    year is excluded from the sample. We proceed analogously in the subsequent columns by excluding
    the longest expiry date up to the point where we perform the calibration on the 2-sample, where
    we fit the smile at 1 and 2 months. We consider market data as of 2nd November 2012. The reference
    exchange rates are $S^{\scriptscriptstyle\mathrm{JPY},\mathrm{AUD}}(0)=83.29$,
    $S^{\scriptscriptstyle\mathrm{USD},\mathrm{AUD}}(0)=1.0375$ and $S^{\scriptscriptstyle\mathrm{JPY},\mathrm{USD}}(0)=80.28$.}
    \label{tab:params_AUDUSDJPY}
\end{table}
\newpage

\begin{figure}[H]
    \centering
        \includegraphics[scale=0.12]{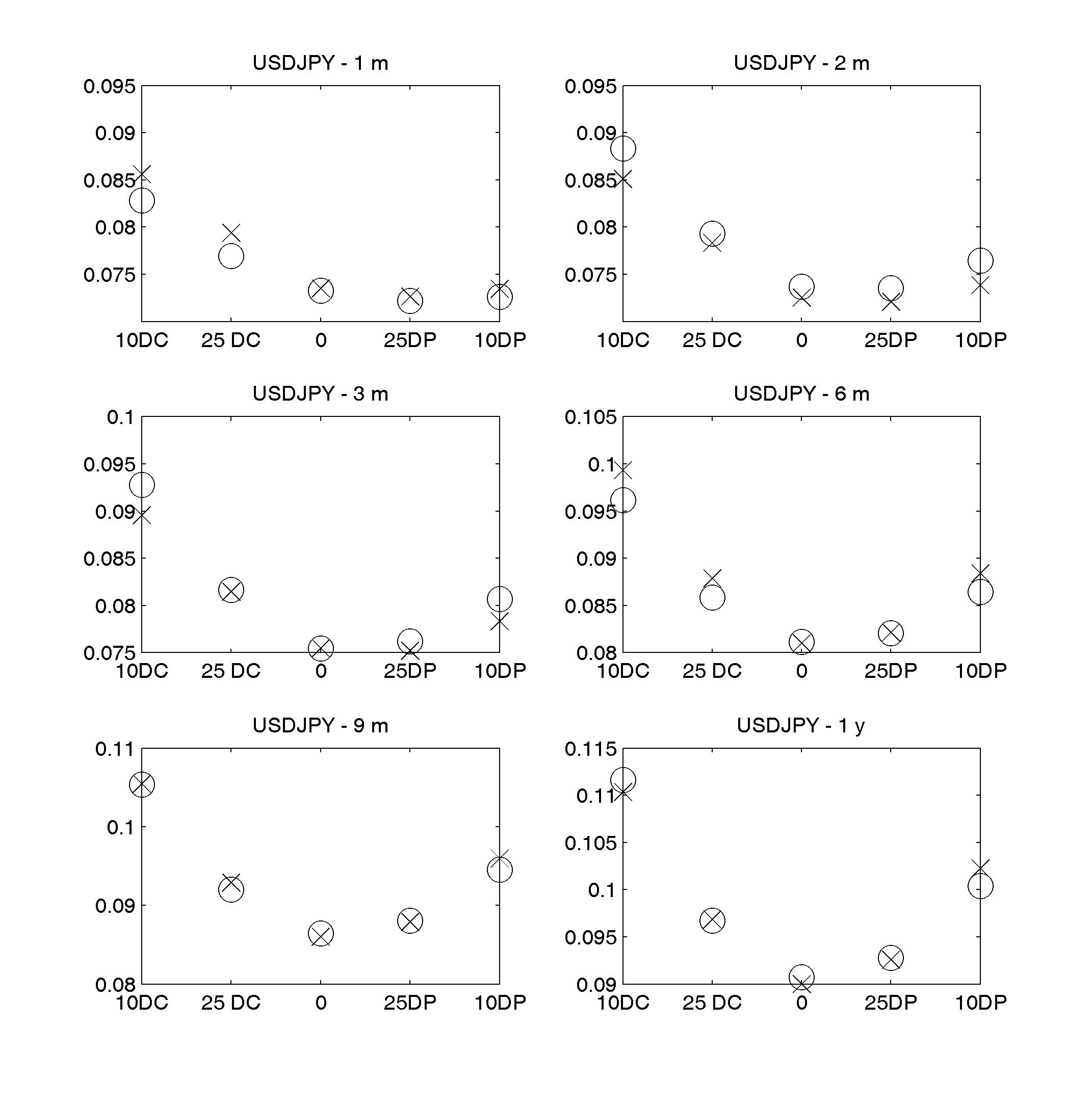}
    \caption{Calibration of the USD/JPY implied volatility surface. Market data as of 2/11/2012. Market volatilities are denoted by crosses,
model volatilities are denoted by circles.  Moneyness levels follow the standard Delta
quoting convention in the FX option market, see Footnote \ref{mktconv}. DC and DP stand for "delta call" and "delta put" respectively.}
    \label{fig:USDJPY6}
\end{figure}

\begin{figure}[H]
    \centering
        \includegraphics[scale=0.12]{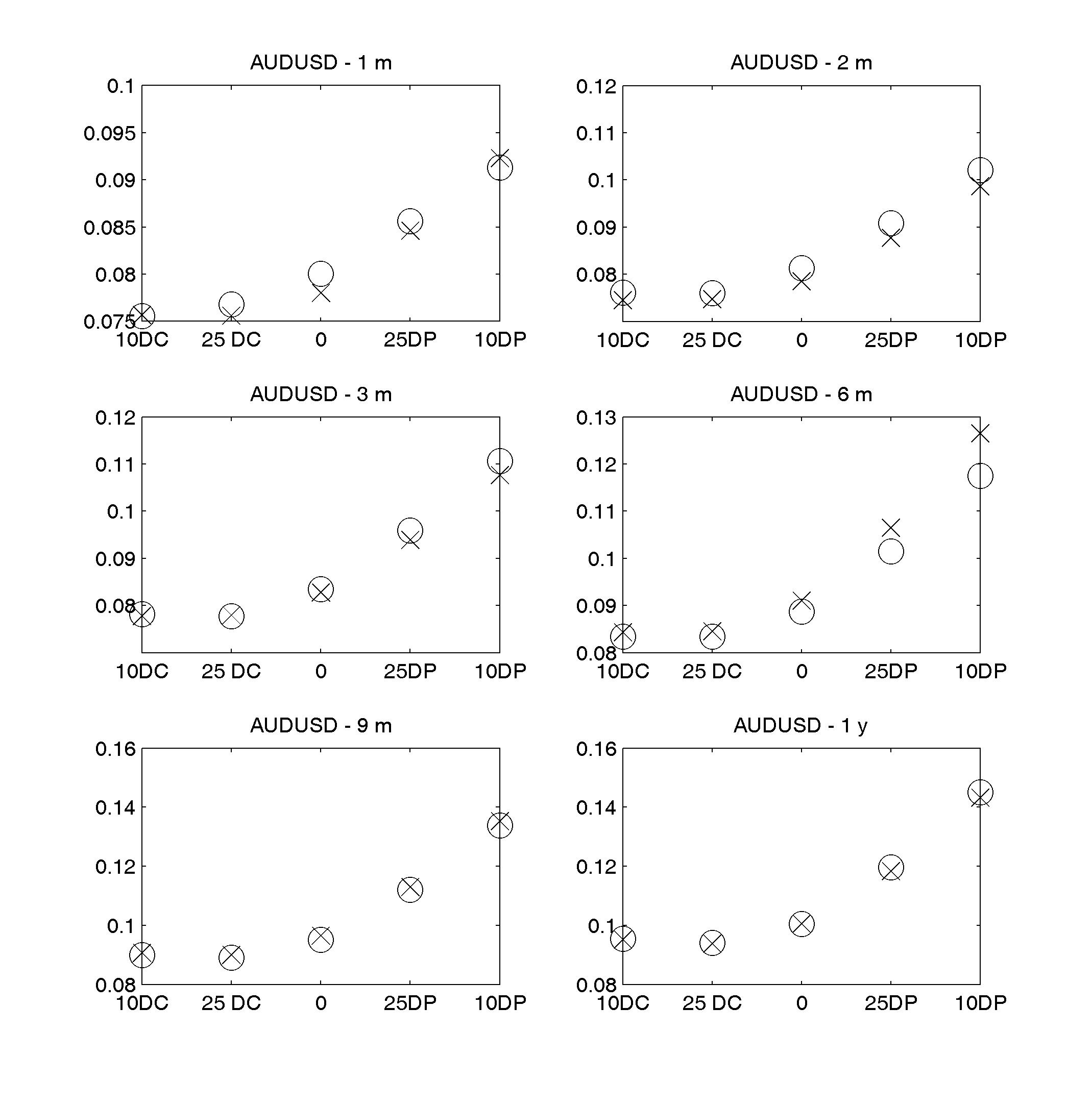}
    \caption{Calibration of the AUD/USD implied volatility surface. Market data as of 2/11/2012. Market volatilities are denoted by crosses,
model volatilities are denoted by circles. Moneyness levels follow the standard Delta
quoting convention in the FX option market, see Footnote \ref{mktconv}. DC and DP stand for "delta call" and "delta put" respectively.}
    \label{fig:AUDUSD6}
\end{figure}

\begin{figure}[H]
    \centering
        \includegraphics[scale=0.12]{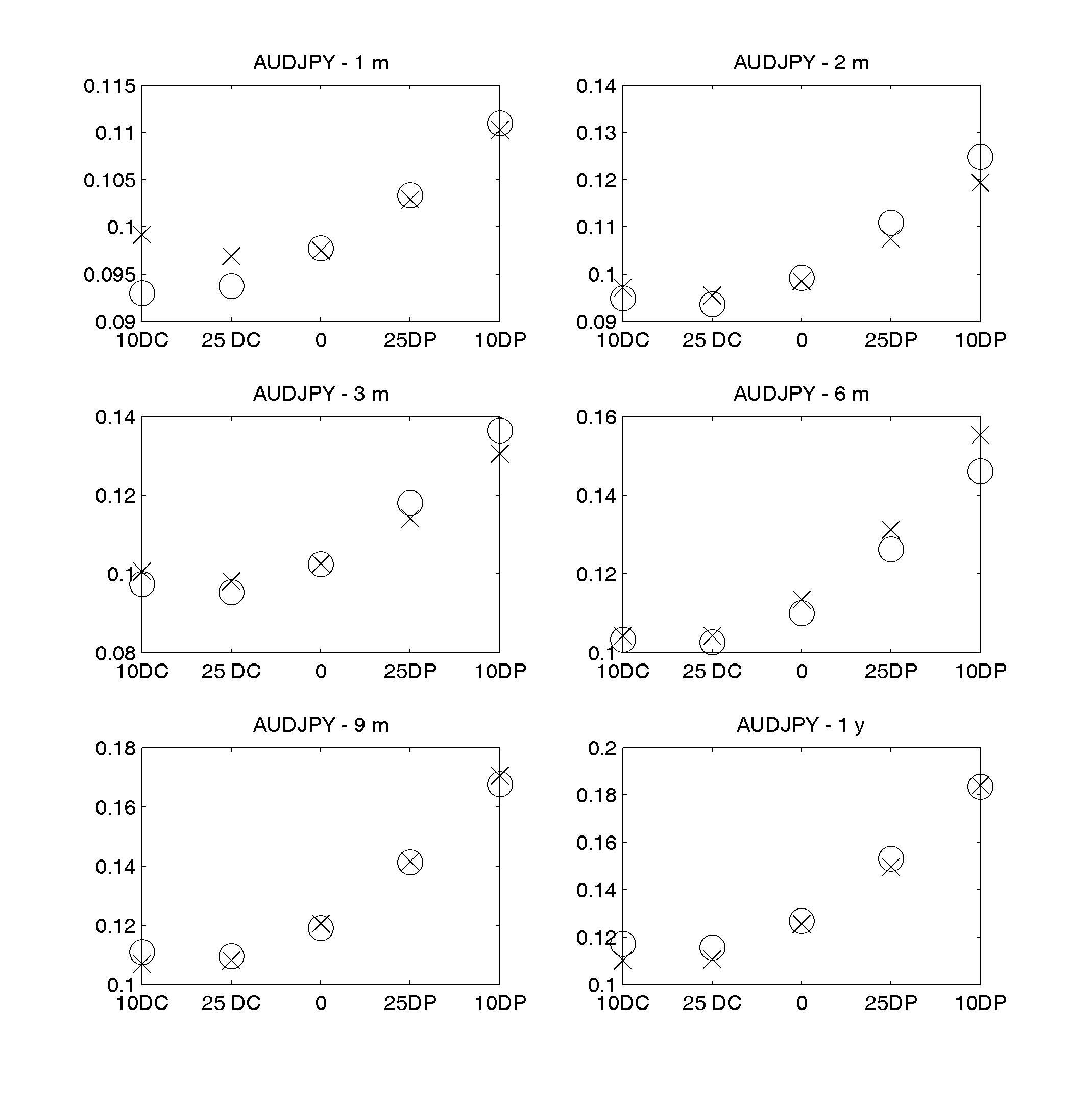}
    \caption{Calibration of the AUD/JPY implied volatility surface. Market data as of 2/11/2012. Market volatilities are denoted by crosses,
model volatilities are denoted by circles. Moneyness levels follow the standard Delta
quoting convention in the FX option market, see Footnote \ref{mktconv}. DC and DP stand for "delta call" and "delta put" respectively.}
    \label{fig:AUDJPY6}
\end{figure}

\begin{figure}[H]
\centering
\subfloat{\label{fig:111}\includegraphics[scale=0.08]{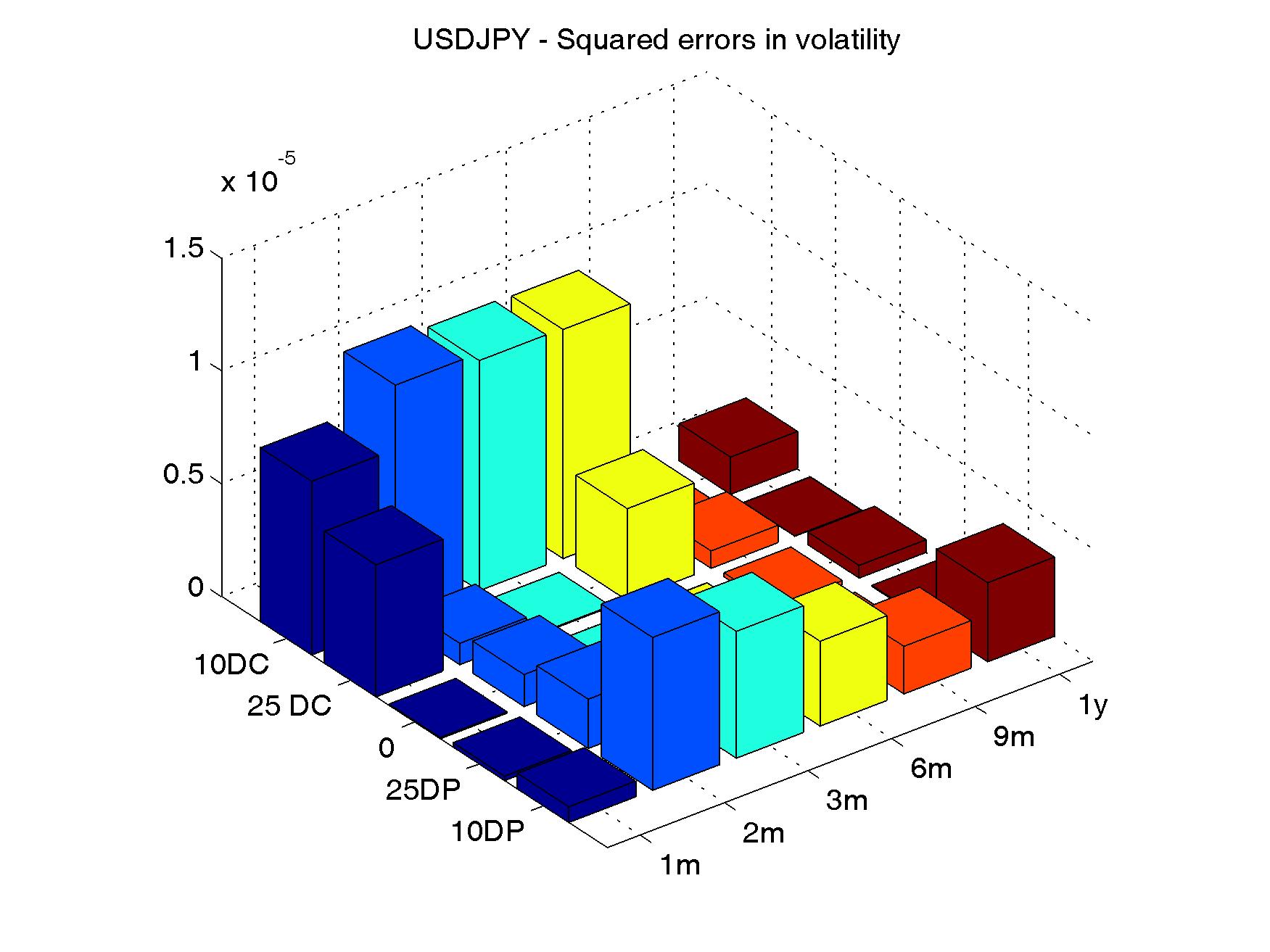}}\\
\subfloat{\label{fig:112}\includegraphics[scale=0.08]{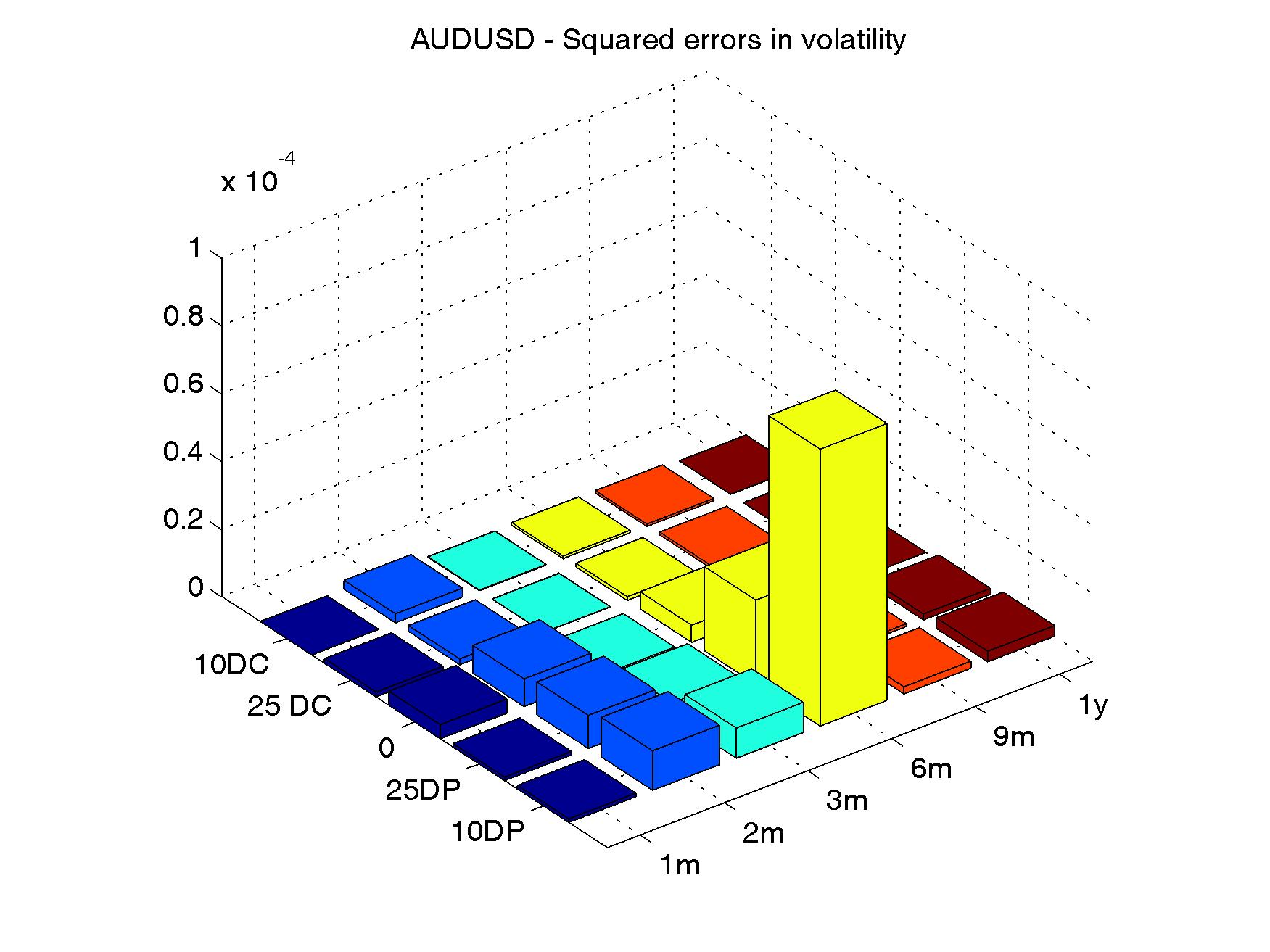}}
\subfloat{\label{fig:113}\includegraphics[scale=0.08]{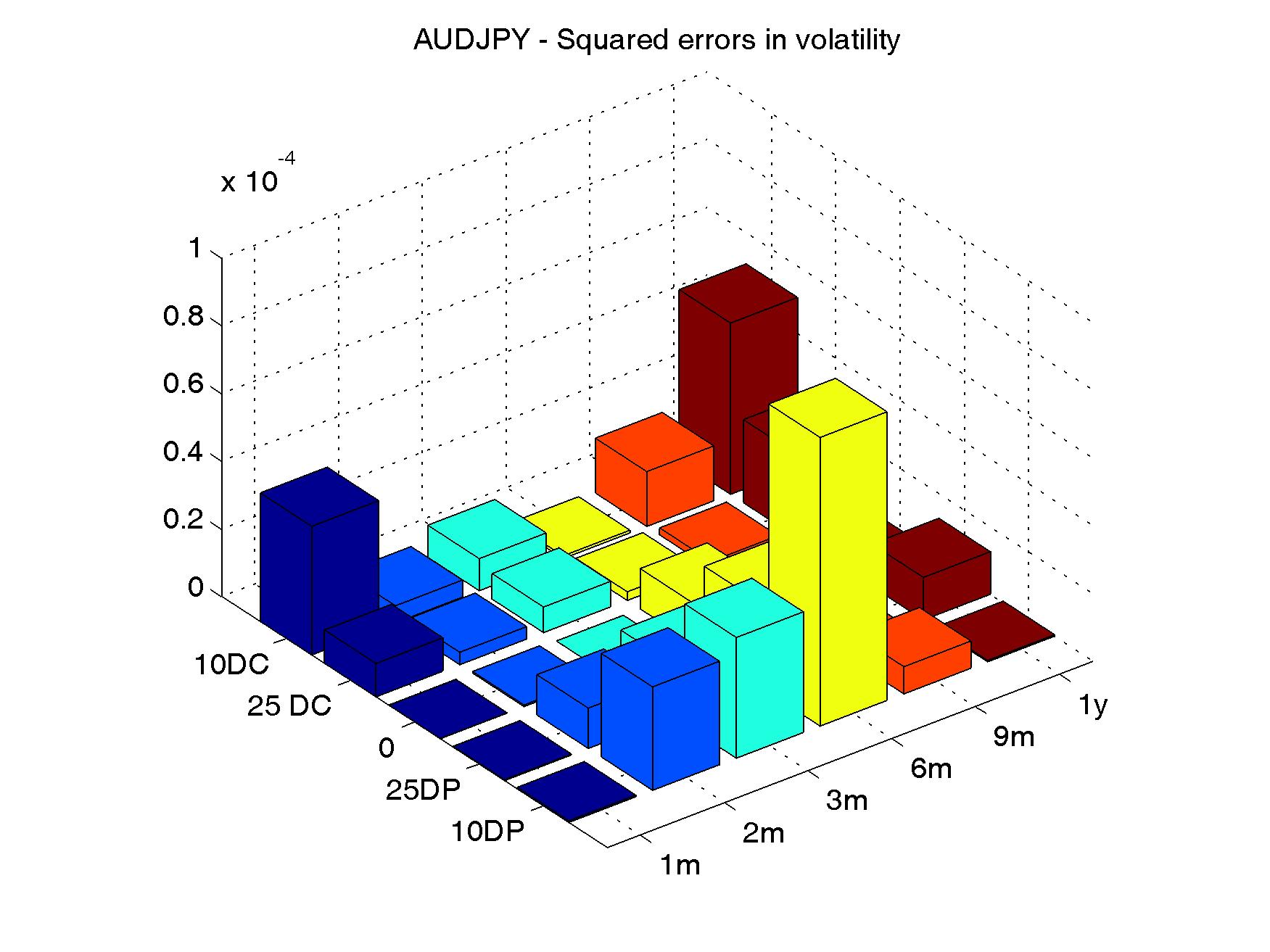}}
\caption{In-sample implied volatilities squared error for the joint calibration on 1m, 2m, 3m, 6m, 9m and 1y. The associated model parameters may be found in Table \ref{tab:params_AUDUSDJPY}, column "6".}
\label{fig:SE_11}
\end{figure}

\newpage

\section{Appendix A: The conditional Laplace transform}

Recall that $\phi\left(\omega,t,\tau,x,\mathbf{V}\right)=G(\mathtt{i}\omega,t,\tau,x,\mathbf{V})$, with $x^{i,j}(t)=\log{S^{i,j}(t)}$. The functions $\phi, G$ represent resp. the characteristic function and the moment generating function of the log-exchange rate. In order to determine these quantities, we first need to write the PDE satisfied by $G$. First of all we write down the dynamics of $x=x^{i,j}$:
\begin{align}
dx(t)&=\left(\left(r^i-r^j\right)-\frac{1}{2}(\mathbf{a}^i-\mathbf{a}^j)^\top\mathrm{Diag}\left(\mathbf{V}(t)\right)\left(\mathbf{a}^i-\mathbf{a}^j\right)\right)dt\nonumber\\
&+\left(\mathbf{a}^i-\mathbf{a}^j\right)^\top\sqrt{\mathrm{Diag}\left(\mathbf{V}(t)\right)}d\mathbf{Z}^{\mathbb{Q}^i}(t).
\end{align}
We also compute the following covariation terms for $k=1,..,d$:
\begin{align}
d\left\langle x,V_{k} \right\rangle_t &=d\left\langle\int_0^.\left(\mathbf{a}^i-\mathbf{a}^j\right)^\top\sqrt{\mathrm{Diag}\left(\mathbf{V}(u)\right)}d\mathbf{Z}^{\mathbb{Q}^i}(u),\int_0^.\xi_k\sqrt{V_k(u)}\rho_kdZ_k^{\mathbb{Q}^i}(u) \right\rangle_t\nonumber\\
&=d\left\langle \int_0^.\left(a^i_k-a^j_k\right)\sqrt{V_k(u)}dZ_k^{\mathbb{Q}^i}(u),\int_0^.\xi_k\sqrt{V_k(u)}\rho_kdZ_k^{\mathbb{Q}^i}(u) \right\rangle_t\nonumber\\
&=\left(a^i_k-a^j_k\right)V_k(t)\xi_k\rho_kdt.
\end{align}
The Laplace transform $G$ solves the following backward Kolmogorov equation \cite{ks91}:
\begin{align}
&-\frac{\partial G}{\partial t}=\frac{1}{2}\frac{\partial^2 G}{\partial x^2}\left(\mathbf{a}^i-\mathbf{a}^j\right)^\top\mathrm{Diag}\left(\mathbf{V}\right)\left(\mathbf{a}^i-\mathbf{a}^j\right)\nonumber\\
&+\sum_{k=1}^{d}{\frac{\partial^2 G}{\partial x\partial V_k}\left(a^i_k-a^j_k\right)V_k\xi_k\rho_k}+\frac{1}{2}\sum_{k=1}^{d}{\frac{\partial^2 G}{\partial V_k^2}\xi_k^2V_k}\nonumber\\
&+\left(\left(r^i-r^j\right)-\frac{1}{2}\left(\mathbf{a}^i-\mathbf{a}^j\right)^\top\mathrm{Diag}\left(\mathbf{V}\right)\left(\mathbf{a}^i-\mathbf{a}^j\right)\right)\frac{\partial G}{\partial x}\nonumber\\
&+\sum_{k=1}^{d}{\frac{\partial G}{\partial V_k}\kappa_k\left(\theta_k-V_k\right)}
\end{align}
with terminal condition $G\left(\omega,T,0,x,\mathbf{V}\right)=e^{\omega x}$ with $\omega\in\mathbb{R}$. In order to solve this problem we look for an exponential affine solution of the form:
\begin{align}
G\left(\omega,t,\tau,x,\mathbf{V}\right)=\exp\left(A(t,T)+\sum_{k=1}^{d}{B_{k}(t,T)V_k}+C(t,T)x\right),
\end{align}
for some deterministic functions $A,B_k,C$ that may depend on both $t,T$. Upon substitution of the guess and recognition of the terms we obtain the following system of $d+2$ ODE's:
\begin{align}
&\frac{\partial A}{\partial t}+\sum_{k=1}^{d}{B_k(t,T)\kappa_k\theta_k}+\left(r^i-r^j\right)C(t,T)=0;\\
&\frac{\partial B_k}{\partial t}+\frac{1}{2}C^2(t,T)\left(a^i_k-a^j_k\right)^2+C(t,T)B_k(t,T)\left(a^i_k-a^j_k\right)\rho_k\xi_k\nonumber\\
&+\frac{1}{2}B_k^2(t,T)\xi_k^{2}-\frac{1}{2}\left(a^i_k-a^j_k\right)^2C(t,T)-B_k(t,T)\kappa_k=0;\\
&\frac{\partial C}{\partial t}=0\label{der1},
\end{align}
with terminal conditions: $ A(T,T)=0, \quad B_k(T,T)=0, \quad C(T,T)=\omega $ for $k=1,..,d$. From \eqref{der1} and its terminal condition, we deduce that $C(t,T)=\omega$ for $t\in[0,T]$, so we can rewrite the system as follows:
\begin{align}
&\frac{\partial A}{\partial t}+\sum_{k=1}^{d}{\kappa_k\theta_kB_k(t,T)}+\left(r^i-r^j\right)\omega=0;\label{eq_A}\\
&\frac{\partial B_k}{\partial t}+\frac{1}{2}B_k^2(t,T)\xi_k^{2}+\left(-\kappa_k+\omega\left(a^i_k-a^j_k\right)\rho_k\xi_k\right)B_k(t,T)\nonumber\\
&+\frac{\omega^2-\omega}{2}\left(a^i_k-a^j_k\right)^2=0,\quad k=1,..,d.\label{Riccati1}
\end{align}
Now for $k=1,..,d$ we assume that $B_k(t,T)$ can be written by means of a function $E_k(t,T)$  and set:
\begin{align}
B_k(t,T)=\frac{\frac{\partial}{\partial t}E_k(t,T)}{\frac{\xi_k^{2}}{2}E_k(t,T)}\label{E1},
\end{align}
then the solution for \eqref{Riccati1} is:
\begin{align}
B_k(t,T)=\frac{\left(\omega^2-\omega\right)}{2}\left(a^i_k-a^j_k\right)^2\frac{1-e^{-\sqrt{\Delta_k}(T-t)}}{\lambda_k^+e^{-\sqrt{\Delta_k}(T-t)}-\lambda_k^-},
\end{align}
with
\begin{align}
\Delta_k&=\left(-\kappa_k+\omega\left(a^i_k-a^j_k\right)\rho_k\xi_k\right)^2-\xi_k^{2}\left(\omega^2-\omega\right)\left(a^i_k-a^j_k\right)^2\\
\lambda_k^\pm&=\frac{\left(-\kappa_k+\omega\left(a^i_k-a^j_k\right)\rho_k\xi_k\right)\pm\sqrt{\Delta_k}}{2}.
\end{align}
Equipped with the solution for $B_k(t,T)$ we can now compute $A(t,T)$ as follows:
\begin{align}
A(T,T)-A(t,T)&=\int_{t}^{T}{\frac{\partial}{\partial u}A(u,T)du}\nonumber\\
A(t,T)&=\int_{t}^{T}{\sum_{k=1}^{d}{\kappa_k\theta_kB_k(u,T)}+\left(r^i-r^j\right)\omega du}\nonumber\\
&=\left(r^i-r^j\right)\omega(T-t)+\sum_{k=1}^{d}{\kappa_k\theta_k\int_{t}^{T}{B_k(u,T)du}}\nonumber\\
&=\left(r^i-r^j\right)\omega(T-t)+\sum_{k=1}^{d}{\frac{2\kappa_k\theta_k}{\xi_k^{2}}\int_{t}^{T}{\frac{\frac{\partial}{\partial t}E_k(t,T)}{E_k(t,T)}du}},
\end{align}
which implies that the solution for $A(t,T)$ is
\begin{align}
A(t,T)&=\left(r^i-r^j\right)\omega(T-t)+\sum_{k=1}^{d}{\frac{2\kappa_k\theta_k}{\xi_k^{2}}\log{\frac{\lambda_k^+-\lambda_k^-}{\lambda_k^+e^{\lambda_k^-(T-t)}-\lambda_k^-e^{\lambda_k^+(T-t)}}}}\nonumber\\
&=\left(r^i-r^j\right)\omega(T-t)+\sum_{k=1}^{d}{A_k(t,T)},
\end{align}
where the functions $A_k(t,T)$ are implicitly defined by the last equality for $k=1,..,d$.
Now we obtain the statement of the proposition once we replace $B^{i,j}_k(\tau)=B_k(t,T),A^{i,j}_k(\tau)=A_k(t,T)$ with $\tau=T-t$.

\section{Appendix B: Expansions}

The calibrations that were presented in Sec.\ \ref{sec:calib} were performed using a deterministic gradient-based optimizers of
the squared distance between the model implied volatilities and market ones.
Model implied volatilities are extracted from the prices produced by the FFT routine. The success of the optimization routine
might be jeopardized by the likely existence of multiple local minima. Hence, it is crucial to start with appropriate initial
guess for the model parameters that are as close as possible to the global minimum. To this aim, approximate solutions for the
option prices are very useful in providing good initial guesses for the parameters.

We present here an approximate expression for the option prices under the multi-Heston model which is asymptotically valid
for small vol-of-vol parameters. The derivation of this formula, which is reported in the next Appendix, relies on arguments which
may be found in \cite{lewis2000} and \cite{dafgra11} (we drop all currency indices, it is intended that we are considering the
$(i,j)$ FX pair).

\begin{proposition}\label{prop_exp1}Assume that all vol-of-vol parameters $\xi_k,k=1,..,d$ have been scaled by the same factor $\alpha>0$.
Then the call price $C(S(t),K,\tau)$ in the Multifactor Heston-based exchange model can be approximated in terms of the scale factor $\alpha$
by differentiating the Black Scholes formula $C_{\scriptscriptstyle\mathrm{BS}}\left(S(t),K,\sigma, \tau\right)$ with respect to the log exchange rate $x(t)=\ln S(t)$
and the integrated variance $v=\sigma^2\tau$:
\begin{align}\label{eq:expansionprice}
        C(S(t),K,\tau)&\approx C_{\scriptscriptstyle\mathrm{BS}}\left(S(t), K, \sigma,\tau\right)\nonumber\\
        &+\alpha\sum_{k=1}^{d}\left({\mathcal{A}_k^{(1)}(\tau)+\mathcal{B}_k^{(1)}(\tau)V_k}\right)
                \partial^2_{xv}C_{\scriptscriptstyle\mathrm{BS}}\left(S(t), K, \sigma,\tau\right)\nonumber\\
        &+\alpha^2\sum_{k=1}^{d}\left({\mathcal{A}_k^{(2)}(\tau)+\mathcal{B}_k^{(2)}(\tau)V_k}\right)
                \partial^2_{vv}C_{\scriptscriptstyle\mathrm{BS}}\left(S(t), K, \sigma,\tau\right)\nonumber\\
        &+\alpha^2\sum_{k=1}^{d}\left({\mathcal{A}_k^{(3)}(\tau)+\mathcal{B}_k^{(3)}(\tau)V_k}\right)
                \partial^3_{xxv}C_{\scriptscriptstyle\mathrm{BS}}\left(S(t), K, \sigma,\tau\right)\nonumber\\
        &+\frac{\alpha^2}{2}\left[\sum_{k=1}^{d}\left({\mathcal{A}_k^{(1)}(\tau)+\mathcal{B}_k^{(1)}(\tau)V_k}\right)\right]^2
                \partial^4_{xxvv}C_{\scriptscriptstyle\mathrm{BS}}\left(S(t), K, \sigma,\tau\right).
\end{align}
We have defined $\tau=T-t$ and the auxiliary real deterministic
functions $\mathcal{B}_k^{(0)},\mathcal{B}_k^{(1)},
\mathcal{B}_k^{(2)},\mathcal{B}_k^{(3)},k=1,..,d$ as
\begin{align}
    \mathcal{B}_k^{(0)}(\tau)&=\left(a^i_k-a^j_k\right)^2\frac{1-e^{-\kappa_k\tau}}{\kappa_k};\label{Bk0}\\
    \mathcal{B}_k^{(1)}(\tau)&= \left(a^i_k-a^j_k\right)^3\rho_k\xi_k\left(\frac{1}{\kappa_k^2}-\frac{e^{-\kappa_k\tau}}{\kappa_k^2}-\frac{\tau e^{-\kappa_k\tau}}{\kappa_k}\right);\\
    \mathcal{B}_k^{(2)}(\tau)    &=\left(a^i_k-a^j_k\right)^4\frac{\xi_k^{2}}{2\kappa_k^2} \left(\frac{1-e^{-2\kappa_k\tau}}{\kappa_k}-2\tau e^{-\kappa_k\tau}\right);\\    \mathcal{B}_k^{(3)}(\tau)&=\left(a^i_k-a^j_k\right)^4\rho_k^{2}\xi_k^{2}\left(\frac{1-e^{-\kappa_k\tau}}{\kappa_k^{3}}-\frac{\tau e^{-\kappa_k\tau}}{\kappa_k^{2}}-\frac{\tau^2e^{-\kappa_k\tau}}{2\kappa_k}\right)\label{Bk21}
\end{align}
and
$\mathcal{A}_k^{(0)},\mathcal{A}_k^{(1)},\mathcal{A}_k^{(2)},\mathcal{A}_k^{(3)},k=1,..,d$
as
\begin{align}
\mathcal{A}_k^{(0)}(\tau)&= \left(a^i_k-a^j_k\right)^2\theta_k\left(\tau +\frac{e^{-\kappa_k\tau}-1}{\kappa_k}\right);\label{mathcal_Ak0}\\
\mathcal{A}_k^{(1)}(\tau)
&= \left(a^i_k-a^j_k\right)^3\theta_k\rho_k\xi_k\left(\frac{\tau}{\kappa_k} +2\frac{e^{-\kappa_k\tau}-1}{\kappa_k^{2}}+\frac{\tau e^{-\kappa_k\tau}}{\kappa_k}\right);\\
\mathcal{A}_k^{(2)}(\tau)
&= \left(a^i_k-a^j_k\right)^4\theta_k\rho_k^2\xi_k^2\left(\frac{\tau}{\kappa_k} +\frac{e^{-\kappa_k\tau}-1}{\kappa_k^3}+\frac{\tau e^{-\kappa_k\tau}}{\kappa_k^2}\right.\nonumber\\
&\left.-\frac{e^{-\kappa_k\tau}-1}{\kappa_k^2}+\frac{\tau^2e^{-\kappa_k\tau}}{2\kappa_k}-\frac{\tau e^{-\kappa_k\tau}}{\kappa_k}+\frac{e^{-\kappa_k\tau}-1}{\kappa_k}\right);\\
\mathcal{A}_k^{(3)}(\tau)
&=\left(a^i_k-a^j_k\right)^4\theta_k\rho_k^2\xi_k^2\left(\frac{\tau}{\kappa_k^2}+3\frac{e^{-\kappa_k\tau}-1}{\kappa_k^3}+2\frac{\tau
e^{-\kappa_k\tau}}{\kappa_k^2}+\frac{\tau^2e^{-\kappa_k\tau}}{2\kappa_k}\right)\label{mathcal_Ak3}.
\end{align}
Finally, the integrated variance reads
\begin{align}
v=\sigma^2\tau=\sum_{k=1}^{d}\bigl({\mathcal{A}_k^{(0)}(\tau)+\mathcal{B}_k^{(0)}(\tau)V_k}\bigr)\label{int_var}.
\end{align}
\end{proposition}
\begin{proof}
The starting point is given by the Riccati ODE \eqref{Riccati1}
expressed in terms of time-to-maturity $\tau=T-t$ and perturbed by
introducing the vol-of-vol scale parameter $\alpha$:
\begin{align}
&\frac{\partial B_k}{\partial \tau}=\frac{1}{2}B_k^2(\tau)\alpha^2\xi_k^{2}+\left(-\kappa_k+\omega\left(a^i_k-a^j_k\right)\rho_k\alpha\xi_k\right)B_k(\tau)\nonumber\\
&+\frac{\omega^2-\omega}{2}\left(a^i_k-a^j_k\right)^2,\quad
k=1,..,d.
\end{align}
We consider the following expansion in terms of $\alpha$:
$B_k(\tau)=B_{k,0}(\tau)+\alpha B_{k,1}(\tau)+\alpha^2
B_{k,2}(\tau)$. By plugging in the expansion and upon recognition
of terms we obtain the following system of ODE's:
\begin{align}
\frac{\partial B_{k,0}}{\partial \tau}&=-\kappa_kB_{k,0}(\tau)+\frac{\omega^2-\omega}{2}\left(a^i_k-a^j_k\right)^2;\\
\frac{\partial B_{k,1}}{\partial \tau}&=-\kappa_kB_{k,1}(\tau)+\omega\left(a^i_k-a^j_k\right)\rho_k\xi_kB_{k,0}(\tau);\\
\frac{\partial B_{k,2}}{\partial
\tau}&=-\kappa_kB_{k,2}(\tau)+\omega\left(a^i_k-a^j_k\right)\rho_k\xi_kB_{k,1}(\tau)+\frac{1}{2}B_{k,0}^2(\tau)\xi_k^{2}.
\end{align}
If we denote $\gamma:=\frac{\omega^2-\omega}{2}$ then the
solutions are easily computed as:
\begin{align}
B_{k,0}(\tau)&=\underbrace{B_{k,0}(0)}_{=0}e^{-\kappa_k\tau}+e^{-\kappa_k\tau}\int_{0}^{\tau}{e^{\kappa_ku}\gamma\left(a^i_k-a^j_k\right)^2du}\nonumber\\
&=\gamma\mathcal{B}_k^{(0)}(\tau);\\
B_{k,1}(\tau)&=\underbrace{B_{k,1}(0)}_{=0}e^{-\kappa_k\tau}+e^{-\kappa_k\tau}\int_{0}^{\tau}{e^{\kappa_ku}\omega\left(a^i_k-a^j_k\right)\rho_k\xi_k\gamma\mathcal{B}_k^{(0)}(u)du}\nonumber\\
&=\omega\gamma\mathcal{B}_k^{(1)}(\tau);\\
B_{k,2}(\tau)&=\underbrace{B_{k,2}(0)}_{=0}e^{-\kappa_k\tau}+e^{-\kappa_k\tau}\int_{0}^{\tau}{e^{\kappa_ku}\omega^2\gamma\left(a^i_k-a^j_k\right)\rho_k\xi_k\mathcal{B}_k^{(1)}(u) du}\nonumber\\
&+e^{-\kappa_k\tau}\int_{0}^{\tau}{e^{\kappa_ku}\gamma\frac{\xi_k^{2}}{2}\left(\mathcal{B}_k^{(0)}(u)\right)^2 du}\nonumber\\
&=\omega^2\gamma\mathcal{B}_k^{(3)}(\tau)+\gamma^2\mathcal{B}_k^{(2)}(\tau).
\end{align}
 Then we can write the function $B_k(\tau)$ as follows:
\begin{align}
B_k(\tau)=\gamma\mathcal{B}_k^{(0)}(\tau)+\alpha\omega\gamma\mathcal{B}_k^{(1)}(\tau)
+\alpha^2\left(\omega^2\gamma\mathcal{B}_k^{(3)}(\tau)+\gamma^2\mathcal{B}_k^{(2)}(\tau)\right).\label{71}
\end{align}
A direct substitution of \eqref{71} into \eqref{eq_A} allows us to
express the function $A(\tau)$:
\begin{align}
A(\tau)&=\omega\left(r^i-r^j\right)\tau+\sum_{k=1}^{d}{\kappa_k\theta_k\int_{0}^{\tau}B_k(u)du}\nonumber\\
&=\omega\left(r^i-r^j\right)\tau+\gamma\sum_{k=1}^{d}\underbrace{\kappa_k\theta_k\int_{0}^{\tau}{\mathcal{B}_k^{(0)}(u)}du}_{:=\mathcal{A}_k^{(0)}(\tau)}+\omega\gamma\alpha\sum_{k=1}^{d}\underbrace{\kappa_k\theta_k\int_{0}^{\tau}{\mathcal{B}_k^{(1)}(u)}du}_{:=\mathcal{A}_k^{(1)}(\tau)}\nonumber\\
&+\omega^2\gamma\alpha^2\sum_{k=1}^{d}\underbrace{\kappa_k\theta_k\int_{0}^{\tau}{\mathcal{B}_k^{(3)}(u)}du}_{:=\mathcal{A}_k^{(3)}(\tau)}+\alpha^2\gamma^2\sum_{k=1}^{d}\underbrace{\kappa_k\theta_k\int_{0}^{\tau}{\mathcal{B}_k^{(2)}(u)}du}_{:=\mathcal{A}_k^{(2)}(\tau)}.
\end{align}
We consider then the price in terms of Fourier transform as in
(\ref{price}) by replacing the argument
$\omega=\mathtt{i}\lambda$. A Taylor-McLaurin expansion w.r.t.
$\alpha$ gives the following:
\begin{align}\label{price_expansion1}
    C(S(t),K,\tau)&\approx\frac{e^{-r^i\tau}}{2\pi}\int_{\mathcal{Z}}{e^{\mathtt{i}\lambda\left(r^i-r^j\right)\tau+\mathtt{i}\lambda x+\gamma\sum_{k=1}^{d}\left({\mathcal{A}_k^{(0)}(\tau)+\mathcal{B}_k^{(0)}(\tau)V_k}\right)}\Phi(\lambda)d\lambda}\nonumber\\
            &+\alpha\sum_{k=1}^{d}\left({\mathcal{A}_k^{(1)}(\tau)+\mathcal{B}_k^{(1)}(\tau)V_k}\right)\nonumber\\
            &\times\frac{e^{-r^i\tau}}{2\pi}\int_{\mathcal{Z}}{\gamma\mathtt{i}\lambda e^{\mathtt{i}\lambda\left(r^i-r^j\right)\tau+\mathtt{i}\lambda x+\gamma\sum_{k=1}^{d}\left({\mathcal{A}_k^{(0)}(\tau)+\mathcal{B}_k^{(0)}(\tau)V_k}\right)}\Phi(\lambda)d\lambda}\nonumber\\
            &+\alpha^2\sum_{k=1}^{d}\left({\mathcal{A}_k^{(2)}(\tau)+\mathcal{B}_k^{(2)}(\tau)V_k}\right)\nonumber\\
            &\times\frac{e^{-r^i\tau}}{2\pi}\int_{\mathcal{Z}}{\gamma^2 e^{\mathtt{i}\lambda\left(r^i-r^j\right)\tau+\mathtt{i}\lambda x+\gamma\sum_{k=1}^{d}\left({\mathcal{A}_k^{(0)}(\tau)+\mathcal{B}_k^{(0)}(\tau)V_k}\right)}\Phi(\lambda)d\lambda}\nonumber\\
            &+\alpha^2\sum_{k=1}^{d}\left({\mathcal{A}_k^{(3)}(\tau)+\mathcal{B}_k^{(3)}(\tau)V_k}\right)\nonumber\\
            &\times\frac{e^{-r^i\tau}}{2\pi}\int_{\mathcal{Z}}{\gamma\mathtt{i}^2\lambda^2 e^{\mathtt{i}\lambda\left(r^i-r^j\right)\tau+\mathtt{i}\lambda x+\gamma\sum_{k=1}^{d}\left({\mathcal{A}_k^{(0)}(\tau)+\mathcal{B}_k^{(0)}(\tau)V_k}\right)}\Phi(\lambda)d\lambda}\nonumber\\
            &+\frac{\alpha^2}{2}\left[\sum_{k=1}^{d}\left({\mathcal{A}_k^{(1)}(\tau)+\mathcal{B}_k^{(1)}(\tau)V_k}\right)\right]^2\nonumber\\
            &\times\frac{e^{-r^i\tau}}{2\pi}\int_{\mathcal{Z}}{\gamma^2\mathtt{i}^2\lambda^2 e^{\mathtt{i}\lambda\left(r^i-r^j\right)\tau+\mathtt{i}\lambda x+\gamma\sum_{k=1}^{d}\left({\mathcal{A}_k^{(0)}(\tau)+\mathcal{B}_k^{(0)}(\tau)V_k}\right)}\Phi(\lambda)d\lambda}.\nonumber\\
\end{align}
Recall now from \eqref{int_var} the definition of the integrated
Black-Scholes variance. In the previous formula, in the first term
on the right hand side, we recognise the Black-Scholes price in
terms of the characteristic function when the integrated variance
is $v = \sigma^2\tau$:
\begin{align}
C_{\scriptscriptstyle\mathrm{BS}}\left(S(t),K,\sigma,\tau\right)=\frac{e^{-r^i\tau}}{2\pi}\int_{\mathcal{Z}}{e^{\mathtt{i}\lambda\left(r^i-r^j\right)\tau+\mathtt{i}\lambda
x+\frac{(\mathtt{i}\lambda)^2-\mathtt{i}\lambda}{2}v}\Phi(\lambda)d\lambda},
\end{align}
so that the price expansion is of the form

\begin{align}
C(S(t),K,\tau)&\approx C_{\scriptscriptstyle\mathrm{BS}}\left(S(t),K,\sigma,\tau\right)\nonumber\\
&+\alpha\sum_{k=1}^{d}\left({\mathcal{A}_k^{(1)}(\tau)+\mathcal{B}_k^{(1)}(\tau)V_k}\right)\partial^2_{xv}C_{\scriptscriptstyle\mathrm{BS}}\left(S(t),K,\sigma,\tau\right)\nonumber\\
&+\alpha^2\sum_{k=1}^{d}\left({\mathcal{A}_k^{(2)}(\tau)+\mathcal{B}_k^{(2)}(\tau)V_k}\right)\partial^2_{vv}C_{\scriptscriptstyle\mathrm{BS}}\left(S(t),K,\sigma,\tau\right)\nonumber\\
&+\alpha^2\sum_{k=1}^{d}\left({\mathcal{A}_k^{(3)}(\tau)+\mathcal{B}_k^{(3)}(\tau)V_k}\right)\partial^3_{xxv}C_{\scriptscriptstyle\mathrm{BS}}\left(S(t),K,\sigma,\tau\right)\nonumber\\
&+\frac{\alpha^2}{2}\left[\sum_{k=1}^{d}\left({\mathcal{A}_k^{(1)}(\tau)+\mathcal{B}_k^{(1)}(\tau)V_k}\right)\right]^2\partial^4_{xxvv}C_{\scriptscriptstyle\mathrm{BS}}\left(S(t),K,\sigma,\tau\right)\label{p_exp_2}.
\end{align}
From the previous expression we can deduce the relation defining
the deterministic functions
$\mathcal{B}_k^{(h)},\mathcal{A}_k^{(h)},\ h=0,...,3$.
\begin{align}
    \mathcal{B}_k^{(0)}(\tau)&=\left(a^i_k-a^j_k\right)^2\frac{1-e^{-\kappa_k\tau}}{\kappa_k};\\
    \mathcal{B}_k^{(1)}(\tau)&=\left(a^i_k-a^j_k\right)\rho_k\xi_k e^{-\kappa_k\tau}\int_{0}^{\tau}{e^{\kappa_ku}\mathcal{B}_k^{(0)}(u)du};\\
    \mathcal{B}_k^{(2)}(\tau)    &=\frac{\xi_k^{2}}{2\kappa_k}e^{-\kappa_k\tau}\int_{0}^{\tau}{e^{\kappa_ku}\left(\mathcal{B}_k^{(0)}(u)\right)^2 du};\\
    \mathcal{B}_k^{(3)}(\tau)&=\left(a^i_k-a^j_k\right)\rho_k\xi_ke^{-\kappa_k\tau}\int_{0}^{\tau}{e^{\kappa_ku}\mathcal{B}_k^{(1)}(u) du}
\end{align}
and
\begin{align}
\mathcal{A}_k^{(0)}(\tau)&=\kappa_k\theta_k\int_{0}^{\tau}{\mathcal{B}_k^{(0)}(u)du};\\
\mathcal{A}_k^{(1)}(\tau)&=\kappa_k\theta_k\int_{0}^{\tau}{\mathcal{B}_k^{(1)}(u)du};\\
\mathcal{A}_k^{(2)}(\tau)&=\kappa_k\theta_k\int_{0}^{\tau}{\mathcal{B}_k^{(2)}(u)du};\\
\mathcal{A}_k^{(3)}(\tau)&=\kappa_k\theta_k\int_{0}^{\tau}{\mathcal{B}_k^{(3)}(u)du}.
\end{align}
Computing the trivial integrals completes the proof.

\end{proof}

We can now present another formula, which does not involve the computation of option prices, and constitutes an approximation
of the implied volatility surface for a short time to maturity. This formula may constitute a useful alternative in order to get a
quicker calibration for short maturities and provides a good initial guess for the parameters in the calibration routine.
The proof is again provided in detail in the next Appendix.

\begin{proposition}\label{prop_exp2}
For a short time to maturity the implied volatility expansion in terms of the vol-of-vol scale factor $\alpha$ in the
multifactor Heston-based exchange model is given by:
\begin{align*}
    \sigma_\mathrm{imp}^2\approx &\ \sigma_0^2+\alpha\left(\sum_{k=1}^{d}{\frac{\rho_k\xi_k}{2}\left(a^i_k-a^j_k\right)^4V_k}\right)
        \frac{m_f}{\sigma_0^2}\nonumber\\
    +&\ \alpha^2\frac{m_f^2}{12\left(\sigma_0^2\right)^2}\left[\sum_{k=1}^{d}\left(1+2\rho_k^2\right)\xi_k^2\left(a^i_k-a^j_k\right)^4V_k
        -\frac{15}{4\sigma_0^2}\left(\sum_{k=1}^{d}{\rho_k\xi_k\left(a^i_k-a^j_k\right)^3V_k}\right)^2\right],
\end{align*}
where $\sigma_0^2=\left(\mathbf{a}^i-\mathbf{a}^j\right)^{\top}\mathrm{Diag}(\mathbf{V})\left(\mathbf{a}^i-\mathbf{a}^j\right)$
and $m_f=\log \left(\frac{S^{i,j}e^{(r^i-r^j)\tau}}{K^{i,j}}\right)$ denotes the forward log-moneyness.\end{proposition}
\begin{proof}
We follow the procedure in \cite{dafgra11}. We suppose an
expansion for the integrated implied variance of the form $v =
\sigma^2_\mathrm{imp}\tau=\zeta_0+\alpha\zeta_1+\alpha^2\zeta_2$
and we consider the Black-Scholes formula as a function of the
integrated implied variance and the log exchange rate $x=\log S$:
$C_{\scriptscriptstyle\mathrm{BS}}\left(S(t),K,\sigma,\tau\right)=C_{\scriptscriptstyle\mathrm{BS}}\left(x(t),K,\sigma_{imp}^2\tau,\tau\right)$.
A Taylor-McLaurin expansion gives us the following:
\begin{align}
C_{\scriptscriptstyle\mathrm{BS}}\left(x(t),K,\sigma_{imp}^2\tau,\tau\right)&=C_{\scriptscriptstyle\mathrm{BS}}\left(x(t),K,\zeta_0,\tau\right)+\alpha\zeta_1\partial_{v}C_{\scriptscriptstyle\mathrm{BS}}\left(x(t),K,\zeta_0,\tau\right)\nonumber\\
&+\frac{\alpha^2}{2}\left(2\zeta_2\partial_v
C_{\scriptscriptstyle\mathrm{BS}}\left(x(t),K,\zeta_0,\tau\right)+\zeta_1^2\partial^2_{v^2}C_{\scriptscriptstyle\mathrm{BS}}\left(x(t),K,\zeta_0,\tau\right)\right).
\end{align}
By comparing this with the price expansion \eqref{p_exp_2} we
deduce that the coefficients must be of the form:
\begin{align}
    \zeta_0&=v_0;\label{zeta_0}\\
    \zeta_1&=\frac{\sum_{k=1}^{d}\left({\mathcal{A}_k^{(1)}(\tau)+\mathcal{B}_k^{(1)}(\tau)V_k}\right)\partial^2_{xv}C_{\scriptscriptstyle\mathrm{BS}}}{\partial_v C_{\scriptscriptstyle\mathrm{BS}}};\label{zeta_1}\\
    \zeta_2&=\frac{-\zeta_1^2\partial^2_{vv} C_{\scriptscriptstyle\mathrm{BS}}+2\sum_{k=1}^{d}\left({\mathcal{A}_k^{(2)}(\tau)+\mathcal{B}_k^{(2)}(\tau)V_k}\right)\partial^2_{vv}C_{\scriptscriptstyle\mathrm{BS}}}{2\partial_v C_{\scriptscriptstyle\mathrm{BS}}}\nonumber\\
        &+ \frac{2\sum_{k=1}^{d}\left({\mathcal{A}_k^{(3)}(\tau)+\mathcal{B}_k^{(3)}(\tau)V_k}\right)\partial^3_{xxv}C_{\scriptscriptstyle\mathrm{BS}}
        +\left[\sum_{k=1}^{d}\left({\mathcal{A}_k^{(1)}(\tau)+\mathcal{B}_k^{(1)}(\tau)V_k}\right)\right]^2\partial^4_{xxvv}C_{\scriptscriptstyle\mathrm{BS}}}
        {2\partial_v C_{\scriptscriptstyle\mathrm{BS}}} \label{zeta_2},
\end{align}
where the Black-Scholes formula
$C_{\scriptscriptstyle\mathrm{BS}}\left(x(t),K,\sigma_{imp}^2\tau,\tau\right)$
is evaluated at the point $\left(x,K,v_0,\tau\right)$. In order to
find the values of $\zeta_1,\zeta_2$, we differentiate
\eqref{Bk0}-\eqref{Bk21} thus obtaining the following ODE's:
\begin{align*}
\frac{\partial \mathcal{B}_k^{(0)}}{\partial \tau}&=-\kappa_k\mathcal{B}_k^{(0)}(\tau)+\left(a^i_k-a^j_k\right)^2;\\
\frac{\partial \mathcal{B}_k^{(1)}}{\partial \tau}&=-\kappa_k\mathcal{B}_k^{(1)}(\tau)+\left(a^i_k-a^j_k\right)\rho_k\xi_k\mathcal{B}_k^{(0)}(\tau);\\
\frac{\partial \mathcal{B}_k^{(2)}}{\partial \tau}&=-\kappa_k\mathcal{B}_k^{(2)}(\tau)+\frac{1}{2}\xi_k^2\mathcal{B}_k^{(0)}(\tau)^2;\\
\frac{\partial \mathcal{B}_k^{(3)}}{\partial
\tau}&=-\kappa_k\mathcal{B}_k^{(3)}(\tau)+\left(a^i_k-a^j_k\right)\rho_k\xi_k\mathcal{B}_k^{(1)}(\tau).
\end{align*}
We consider a Taylor-McLaurin expansion in terms of $\tau$:
\begin{align}
\mathcal{B}_k^{(0)}(\tau)&=\left(a^i_k-a^j_k\right)^2\tau-\frac{\tau^2}{2}\kappa_k\left(a^i_k-a^j_k\right)^2;\\
\mathcal{B}_k^{(1)}(\tau)&=\frac{\tau^2}{2}\left(a^i_k-a^j_k\right)^3\rho_k\xi_k-\frac{2}{3}\tau^3\kappa_k\left(a^i_k-a^j_k\right)^3\rho^k\xi^k;\\
\mathcal{B}_k^{(2)}(\tau)&=\frac{\tau^3}{6}\xi_k^2\left(a^i_k-a^j_k\right)^4;\\
\mathcal{B}_k^{(3)}(\tau)&=\frac{\tau^3}{6}\left(a^i_k-a^j_k\right)^4\rho_k^2\xi_k^2.
\end{align}
Noting from \eqref{mathcal_Ak0}-\eqref{mathcal_Ak3} that
$\mathcal{A}_k^{(i)}$ are one order in $\tau$ higher than the
corresponding $\mathcal{B}_{k}^{(i)}$, the following
approximations hold:
\begin{align}
\sum_{k=1}^{d}\left({\mathcal{A}_k^{(0)}(\tau)+\mathcal{B}_k^{(0)}(\tau)V_k}\right)&=\sum_{k=1}^{d}{\left(a^i_k-a^j_k\right)^2V_k}\tau+o(\tau);\label{start2}\\
\sum_{k=1}^{d}\left({\mathcal{A}_k^{(1)}(\tau)+\mathcal{B}_k^{(1)}(\tau)V_k}\right)&=\sum_{k=1}^{d}{\rho_k\xi_k\left(a^i_k-a^j_k\right)^3V_k}\frac{\tau^2}{2}+o(\tau^2);\\
\sum_{k=1}^{d}\left({\mathcal{A}_k^{(2)}(\tau)+\mathcal{B}_k^{(2)}(\tau)V_k}\right)&=\sum_{k=1}^{d}{\xi_k^2\left(a^i_k-a^j_k\right)^4V_k}\frac{\tau^3}{6}+o(\tau^3);\\
\sum_{k=1}^{d}\left({\mathcal{A}_k^{(3)}(\tau)+\mathcal{B}_k^{(3)}(\tau)V_k}\right)&=\sum_{k=1}^{d}{\rho_k^2\xi_k^2\left(a^i_k-a^j_k\right)^4V_k\frac{\tau^3}{6}}+o(\tau^3).
\end{align}
We introduce two variables: the log-forward moneyness $m_f=\log\left(\frac{Se^{(r^i-r^j)\tau}}{K}\right)$ and \\
$V=\left(\mathbf{a}^i-\mathbf{a}^j\right)^{\top}\mathrm{Diag}(\mathbf{V})\left(\mathbf{a}^i-\mathbf{a}^j\right)\tau$.
Then, from \cite{lewis2000}, we consider the following ratios
among the derivatives of the Black-Scholes formula:
\begin{align}
\frac{\partial^2_{xv}C_{\scriptscriptstyle\mathrm{BS}}\left(x,K,V,\tau\right)}{\partial_{v}C_{\scriptscriptstyle\mathrm{BS}}\left(x,K,V,\tau\right)}&=\frac{1}{2}+\frac{m_f}{V};\\
\frac{\partial^2_{vv}C_{\scriptscriptstyle\mathrm{BS}}\left(x,K,V,\tau\right)}{\partial_{v}C_{\scriptscriptstyle\mathrm{BS}}\left(x,K,V,\tau\right)}&=\frac{m_f^2}{2V^2}-\frac{1}{2V}-\frac{1}{8};\\
\frac{\partial^3_{xxv}C_{\scriptscriptstyle\mathrm{BS}}\left(x,K,V,\tau\right)}{\partial_{v}C_{\scriptscriptstyle\mathrm{BS}}\left(x,K,V,\tau\right)}&=\frac{1}{4}+\frac{m_f-1}{V}+\frac{m_f^2}{V^2};\\
\frac{\partial^4_{xxvv}C_{\scriptscriptstyle\mathrm{BS}}\left(x,K,V,\tau\right)}{\partial_{v}C_{\scriptscriptstyle\mathrm{BS}}\left(x,K,V,\tau\right)}&=\frac{m_f^4}{2V^4}+\frac{m_f^2\left(m_f-1\right)}{2V^3}.\label{end2}
\end{align}
Upon substitution of \eqref{start2}-\eqref{end2} into
\eqref{zeta_0}-\eqref{zeta_2}, we obtain the values for
$\zeta_i,i=0,1,2$ allowing us to express the expansion of the
implied volatility.
\begin{align}
\zeta_0&=\left(\mathbf{a}^i-\mathbf{a}^j\right)^{\top}\mathrm{Diag}(\mathbf{V})\left(\mathbf{a}^i-\mathbf{a}^j\right)\tau;\\
\zeta_1&=\left(\sum_{k=1}^{d}{\frac{\rho_k\xi_k}{2}\left(a^i_k-a^j_k\right)^3V_k}\right)\frac{m_f}{\left(\mathbf{a}^i-\mathbf{a}^j\right)^{\top}\mathrm{Diag}(\mathbf{V})\left(\mathbf{a}^i-\mathbf{a}^j\right)}\tau;\\
\zeta_2&=\frac{m_f^2}{\left(\left(\mathbf{a}^i-\mathbf{a}^j\right)^{\top}\mathrm{Diag}(\mathbf{V})\left(\mathbf{a}^i-\mathbf{a}^j\right)\right)^2}\tau\Bigg[\frac{1}{12}\left(\sum_{k=1}^{d}{\xi_k^2\left(a^i_k-a^j_k\right)^4V_k}\right)\Bigg.\nonumber\\
&\Bigg.+\frac{1}{6}\left(\sum_{k=1}^{d}{\rho_k^2\xi_k^2\left(a^i_k-a^j_k\right)^4V_k}\right)-\frac{5}{16}\frac{\left(\sum_{k=1}^{d}{\rho_{k}\xi_{k}\left(a^i_k-a^j_k\right)^3V_k}\right)^2}{\left(\mathbf{a}^i-\mathbf{a}^j\right)^{\top}\mathrm{Diag}(\mathbf{V})\left(\mathbf{a}^i-\mathbf{a}^j\right)}\Bigg].
\end{align}

\end{proof}
\newpage
\bibliography{biblio}
\bibliographystyle{apalike}
\bibliographystyle{apa}

\end{document}